%% file: lphdregold.tex
\documentclass[journal,12pt]{article}

\RequirePackage[OT1]{fontenc}
\RequirePackage{amsmath,amssymb,amsthm}
\RequirePackage{natbib}
\RequirePackage[colorlinks,citecolor=blue,urlcolor=blue]{hyperref}

\numberwithin{equation}{section}
\theoremstyle{plain}

\usepackage{fancyvrb}
\usepackage{pdfpages,caption}
\usepackage{csquotes}
\usepackage{natbib}
\usepackage{multirow}
\usepackage{array,graphicx}
\usepackage{morefloats}
\usepackage{wrapfig}
\usepackage{savesym}
\usepackage{enumerate}
\usepackage{paralist}
\usepackage{geometry}

\usepackage{booktabs}
\usepackage{multirow}
\usepackage{caption}
\usepackage{subcaption}
\usepackage{rotating}
\usepackage{soul}
\usepackage{bm}  
\usepackage[colorinlistoftodos]{todonotes} 
\newcommand{\ignore}[1]{}

\newtheorem{theorem}{Theorem}[section]
\newtheorem{lemma}[theorem]{Lemma}

\newtheorem{result}[theorem]{Result}
\newtheorem{algorithm}[theorem]{Algorithm}

\newenvironment{remark}[1][Remark]{\begin{trivlist}
\item[\hskip \labelsep {\bfseries #1}]}{\end{trivlist}}
\newcommand{\PreserveBackslash}[1]{\let\temp=\\#1\let\\=\temp}
\newcolumntype{R}[1]{>{\PreserveBackslash\raggedleft}p{#1}}

\newcommand{\abs}[1]{\left|{#1}\right|}
\newcommand{\norm}[1]{\left\lVert#1\right\rVert}
\newcommand*{\tran}[1]{#1^{\mkern-1.5mu\mathsf{T}}}

\newcommand{\IE}{\mathbb{E}}

\newcommand{\IR}{\mathbb{R}}

\newcommand{\IP}{\mathbb{P}}

\def\argmin{\mathop{\rm argmin}}

\begin{document}

\title{Long-term prediction intervals with many covariates}

\author{Sayar~Karmakar\footnote{Corresponding author, email: sayarkarmakar@ufl.edu, Address: 230 Newell Drive, Gainesville, FL- 32601, USA. Phone: +1(352)273-1839.},
        Marek~Chud\'y,
        and~Wei~Biao~Wu \\
        University of Florida, University of Vienna, University of Chicago
}


\maketitle

\begin{abstract}
Accurate forecasting is one of the fundamental focuses in the literature of econometric time-series. Often practitioners and policymakers want to predict outcomes of an entire time horizon in the future instead of just a single $k$-step ahead prediction. These series, apart from their own possible non-linear dependence, are often also influenced by many external predictors. In this paper, we construct prediction intervals of time-aggregated forecasts in a high-dimensional regression setting. Our approach is based on quantiles of residuals obtained by the popular LASSO routine. We allow for general heavy-tailed, long-memory, and nonlinear stationary error processes and stochastic predictors. Through a series of systematically arranged consistency results, we provide theoretical guarantees of our proposed quantile-based method in all of these scenarios. After validating our approach using simulations we also propose a novel bootstrap-based method that can boost the coverage of the theoretical intervals. Finally analyzing the EPEX Spot data, we construct prediction intervals for hourly electricity prices over horizons spanning 17 weeks and contrast them to selected Bayesian and bootstrap interval forecasts.
\end{abstract}


\section{Introduction}\label{sec:intro}

\noindent Prediction intervals (PI hereafter) help forecasters to access the uncertainty concerning the future values of time series. This has been a central topic of interest in the analysis of time series from both theoretical and algorithmic point of view. The $k$-step ahead prediction of a time-series after observing $X_1, \ldots, X_t$ has been discussed from many different perspectives. While estimation of the conditional mean $\IE(X_{n+k}|X_1,\ldots,X_n)$ has received most of the attention since \cite{11}; other approaches were also developed significantly in the last few decades, see \cite{20}, \cite{29}, \cite{18}, \cite{48}, \cite{31} \cite{33}, \cite{8}, \cite{17}, \cite{35}. Nonetheless in the eye of an applied scientist, evaluation of prediction intervals ( PIs henceforth) is challenging \citep[e.g.,][]{c93,ct03}, since a higher empirical coverage probability often comes with the cost of larger width, thus less precision. In this paper, we focus on the construction and evaluation of PIs for a single univariate series using information from many predictors under a possibly high-dimensional linear model with a general error process. Moreover, the target for which we construct the PIs is not a single future value of the univariate series, but a sum of its values over the entire forecasting horizon. This setup has practical applications in telecommunication and data services, energy, and finance. The characteristic feature of these sectors is that data about sales, prices, or returns are collected with hourly frequency while the management is interested in predictions of aggregated volumes over horizon spanning several weeks or months. In the volatility forecasting of stock returns using popular ARCH/GARCH models, often aggregated (over time) mean squared error (AMSE) is an evaluation metric (See \citep{starica2003garch,subbarao2008,sayar2020} etc.).  When it comes to predicting energy or electricity consumption, policy makers want to predict future usages for entire one month or several months ahead.  The aggregated predictions help them in setting the price of such commodities. See some of the applications mentioned in \citet{zhou10}. Time-aggregation is also useful to arrive at strategic decisions made by trust funds, pension management and insurance companies, portfolio management of specific derivatives \citep{kw13} and assets \citep[see][]{bky16} among others.

Towards including external regressors, \citet{lfn15} found that the inclusion of disaggregated wind speed and temperature (measured at more than 70 weather stations across Germany) leads to improvements of forecasts for EPEX~SPOT hourly day-ahead electricity prices.\footnote{In electricity price forecasting (EPF), one has to control for weather conditions, local economy, and environmental policy \citep{kr05,hrz12}. Additionally, EPF is challenging due to complex seasonality (daily, weekly, and yearly), heteroscedasticity, heavy-tails, and sudden price spikes \citep[see][for a recent review]{w14}.}.  We adapt their framework to allow for long-run dependence which is important for long-horizon and medium-horizon forecasting. Similar applications can be found in in power portfolio risk management, derivatives pricing, medium-term and long-term contract evaluation, and maintenance scheduling. To challenge our forecasts, we provide out-of-sample comparison with several forecasting methods, which include Bayes PIs of \citet{mw16} and bootstrap PIs obtained from methods such as exponential smoothing, neural networks, and regression with auto-correlated errors implemented in the R-package \texttt{forecast} \citep[see][]{hk08}. 

An interesting advantage of predicting a sum of future values instead of just a single $k$-step ahead prediction is it allows for some weak law of large numbers to kick in if the horizon length of prediction is large. Towards the methodological development and corresponding theoretical advancement, we explore the properties of the quantile-based PIs of \citet{zhou10} in a high-dimensional regression setting. First we identify that \citet{zhou10} uses linear processes in their modelling of innovation processes. This limits the scope of applicability. Using the idea of predictive density and the functional dependence framework as proposed in a seminal paper by  \citet{wu05}, we extend the quantile consistency results to a large class of nonlinear processes which includes thresholded/transition autoregressive \citep{ltd03} processes and some nonlinear version of GARCH processes. 

On the regression front, the quantile consistency results obtained so far were limited to only a low-dimension regression setting. This builds a strong motivation to explore a high-dimensional regression setting as this would allow us to incorporate several predictors that better explain the variability in the data. In particular, we put a special focus on the scenario where the number of predictors grows much faster than the sample size ($\log p=o(n)$). Using the popular LASSO estimator, we were able to provide theoretical guarantees for these PIs. One interesting contribution of this paper is how our consistency rates explicitly depict the price of having short or long-range dependence and lighter or heavier tails. We also show sharp consistency results for stochastic design in the high-dimensional regime which can be seen as a contribution that is important on its own. Thus on one hand, this paper advances the theory of future-aggregated prediction by enlarging the scope of time-series; at the same time, it allows for many covariates in a high-dimensional regression setting. In particular, this paper can be seen to provide some theoretical justification for how a simple LASSO-based prediction routine can sustain such general time-series dependence. 

Besides the theoretical contributions, we tackle the practical validity of PIs in the scenario of long-horizon forecasting with only a relatively short sample. One new discovery was that the original quantile-based PIs fail when the length of time-horizon to be aggregated for forecasting grows compared to the length of the past data. In light of this shortcoming, we employ a bootstrap-assisted step to improve the out-of-sample coverage probability. Apart from a conjectural viewpoint justifying our bootstrap procedure, these simulation results support our choice of PIs for the final out-of-sample experiment using real-world data on hourly electricity prices. 


The rest of the paper is organized as follows:  Section \ref{sec:methods} shows the construction of quantile-based PIs and also details a novel bootstrap adjustment to improve coverage. Section \ref{sec:noregress} states the quantile consistency results in an exhaustive number of cases: the error/innovation process can be linear or non-linear, short-range or long-range and can have light or heavy tail. Building on the consistency for just the error process, we use Section \ref{sec:high dimension} to consider consistency results for LASSO-fitted residuals. Apart from the traditional fixed design, we also provide discussion on stochastic design in this section. Moving on, the simulations are shown in details in Section \ref{sec:simu} where we compare OLS, LAD and LASSO for both low-dimension and high-dimension regression setup. Section \ref{sec:real data analysis} is used to analyze EPEX spot electricity using several methods including ours. For contrasting our method with the competing ones in literature we use a Pseudo-out-of-sample (POOS hereafter) approach. Finally Section \ref{sec:discussion} provides concluding remarks. We defer some theoretical results and all proofs of theorems to the appendix. 

We now introduce some notation. For a random vector $Y$, write $Y \in \mathcal{L}_p$, for $p > 0$, if $\|Y \|_p := E(|Y |^p ) ^{1/p} < \infty$.  For the $\mathcal{L}_2$ norm write $\|\cdot \| =\| \cdot \|_ 2$. Throughout the text, $c_p$ denotes a constant that depends only on $p$ and $c$ denotes a universal constants. These might take different values in different lines, unless otherwise specified. Then, $x^+=\max(x,0)$ and $x^{-}=-\min(x,0)$. For two positive sequences $a_n$ and $b_n$, if $a_n/b_n \to 0$, write $a_n =O(b_n)$. Write $a_n \lesssim b_n$ if $a_n \leq c b_n$, for some $c<\infty$. For a random variable sequence $X_n$ and a possibly random or non-random sequence $Y_n$, we say $X_n=o_{\IP}(Y_n)$ if $X_n/Y_n \to 0$ in probability. Respectively if $X_n/Y_n$ is a tight random variable we say $X_n=O_{\IP}(Y_n)$. The $d$-variate normal distribution with mean $\mu$ and covariance matrix $\Sigma$ is denoted by $N(\mu, \Sigma)$. Denote by $I_d$ the $d \times d$ identity matrix. For a matrix $A= (a_{ij})$, we define its element-wise $\ell-\infty$ norm as $|A|_{\infty}= max_{ij} |a_{ij}|$ and its $\ell_1$ norm as $|A|_1= max_{j} \sum_i |a_{ij}|$.

\section{Methods: Construction of prediction intervals}\label{sec:methods}
\noindent Suppose an univariate target time series $y_i,i=1,\ldots,n$ follows a regression model 
\begin{equation}\label{eq:regression}
y_i=\tran{\bm{x}}_i  \bm \beta +e_i, \quad \bm \beta \in \IR^p, 
\end{equation}
where $e_i$ stands for the mean-zero temporally-dependent error process. Assuming model \eqref{eq:regression}, we wish to construct PI for $y_{n+1}+\ldots+y_{n+m}$ after observing $(y_i,\bm x_i);i=1,\ldots,n$. We first discuss the scenario $\bm \beta=0$, i.e., $y_i=e_i$ is a zero-mean noise process. This serves as a primer to the high-dimensional regression problem and introduces a basic quantile-based method of constructing the prediction interval.

\subsection{Without covariates}\label{ssc:nocov}
\noindent We will appropriately define our notion of short-range and long-range and light-tailed and heavy-tailed distribution later in Section \ref{sec:nonlinear}. Loosely speaking, short-range dependence stands for faster decay of dependence as the lag increases where long-range dependent process kills dependence slowly. Heavy-tailed refers to the scenario where the error process does not possess even the finite second moment. Depending on the strength of dependence and the tail behavior, \citet{zhou10} proposed two different types of PIs for $m$-step ahead aggregated response $e_{n+1}+\ldots+e_{n+m}$. We present them here, with some suggestive modifications especially for the first approach where we estimate the long-run variance differently. 
\subsubsection{Quenched CLT method}
If the process $e_t$ shows short-range dependence and light-tailed behavior, then in the light of a quenched central limit theorem, \citet{zhou10} proposed the following PI for $\frac{1}{m}(e_{n+1}+\cdots+e_{n+m})$,
    \begin{eqnarray}
    [L,U]=[Q^N(\alpha/2),Q^N(1-\alpha/2)]\frac{\sigma}{\sqrt{m}},
    \end{eqnarray}
    \noindent where $\sigma$ is the long-run standard deviation (sd) of $e_t$. However, since $\sigma$ is unknown, it must be estimated. One can estimate the long-run variance $\sigma^2$ of the $e_i$ process using the sub-sampling block estimator \citep[see eq. (2) in][]{dfsvw13}
\begin{equation}\label{eq:ss estimator}
\tilde{\sigma}=\frac{\sqrt{\pi l/2}}{n}\sum_{k=1}^{\kappa}\abs{\sum_{i=(k-1)l+1}^{kl}e_i},
\end{equation}
with block length $l$ and number of blocks $\kappa=\lceil n/l\rceil$ in order to obtain $100(1-\alpha)\%$ asymptotic PI
$$[L,U]=\pm \hat \sigma Q^t_{\kappa-1}(\alpha/2) \sqrt{m},$$
where $Q^t_{\kappa-1}$ is student-t quantile with $\kappa-1$ degrees of freedom. 

\ignore{
\todo{Do we keep this description here?}
The justification behind such an interval will be discussed in detail as part of our asymptotic results where we show convergence of $(e_{n+1}+\ldots+e_{n+m})/m$ to normal distribution under mild assumptions. 
}

\subsubsection{Empirical method based on quantiles}
A substantially more general method that can account for long-range dependence or heavy-tailed behavior of the error process uses the following quantile  
\begin{eqnarray*}
\hat{Q}(u), \text{ the $u$-th empirical quantile of } \sum_{j=i-m+1}^{i}e_i; i=m,\ldots, n.
\end{eqnarray*}
The PI in this case is
\begin{equation}\label{eq:EQpis}
[L,U]= \left[\hat{Q}(\alpha/2) , \hat{Q}(1-\alpha)/2)\right].
\end{equation}
This approach enjoys reasonable coverage for moderate rate of growth for $m$ compared to sample size $n$.
\subsection{Low dimensional regression}\label{ssc:withcov}
\noindent Now assume the scenario where $\bm \beta$ is possibly non-zero in (\ref{eq:regression}). For $p<n$, with dependent but linear error process, after estimating $\bm{\bm \beta}$, \citet{zhou10} constructed PI as  
\begin{equation}\label{eq:pi}
\sum_{i=n+1}^{n+m} \tran{\bm{x}}_i  \hat{\bm{\bm \beta}}+\text{ PI for } \sum_{i=n+1}^{n+m} \hat{e}_i,
\end{equation} 
where $\hat{e}_i=y_i-\tran{\bm{x}}_i\hat{\bm{\bm \beta}}$ are the regression residuals. Regarding the choice of estimator for $\bm{\bm \beta}$, if the error process shows light tailed behavior and short-range dependence, we typically use OLS estimator $ \bm{\hat{\bm \beta}}=\argmin \sum_i (y_i-\tran{\bm{x}}_i \bm{\bm \beta})^2 $. For heavy-tailed or long-range dependent errors \citep[see][]{huber}, it is better  to use robust regression with general distance $\rho(\cdot)$ and 
$\bm{\hat{\bm \beta}}= \argmin \sum_i \rho (y_i-\tran{\bm{x}}_i \bm{\bm \beta}).$ Examples of distance include the $\mathcal{L}^q$ regression for $1 \leq q \leq 2$. In this paper our focus is on a LASSO-type least square estimation (cf. \ref{eq:lasso objective}) in presence of possibly dependent and nonlinear errors where the number of covariates $p$ is much larger than the sample size $n$.



\subsection{High dimensional regression}\label{ssc:hd}
 Consider the high dimensional regression situation i.e. $p \gg n$. Here we use the popular LASSO estimator
\begin{equation}\label{eq:lasso objective}
\hat{\bm{\bm \beta}} = \argmin_{\bm{\bm \beta} \in \mathcal{R}^p} \frac{1}{n}\sum_{i=1}^n(y_i-\tran{\bm{x}}_i \bm{\bm \beta})^2+ \lambda \sum_{j=1}^p \abs{\beta_j},
\end{equation} 
with the penalty coefficient $\lambda$. Then we get PIs (\ref{eq:pi}) with $\bm{\hat{\bm \beta}}$ replaced by the LASSO estimator. One of the key contributions of our paper lies in the fact that the intuitive and computationally fast lasso estimator adapts well for the time-aggregated response scenario but still provides theoretical consistency guarantee as we show in Section \ref{sec:high dimension}.

\subsection{Future predictors}
Note that, the PI in (\ref{eq:pi}) requires the future values of covariates  $\bm x_i$ namely $\bm{x}_{n+1},\ldots, \bm{x}_{n+m}$. \citet{zhou10} discussed the scenario where these predictors are trigonometric and thus perfectly predictable. For practical implementation in econometrics, it is more appealing to consider stochastic design so that one can allow for the covariates $\bm x_i$ to be observed random variables as well. This naturally raises the question of how to predict the future values of the covariates. For the theoretical part, we shed some light on consistency results for the stochastic design in Section \ref{sec:high dimension}. However, in the implementation in Section \ref{sec:real data analysis}, we choose to simply use a year-over-year mean forecast for the weather covariates.

\subsection{Bootstrap adjustment}\label{ssc:bootstrap}
In our actual implementation for real data, we propose a bootstrap adjusted version of the algorithm in Section \ref{ssc:withcov}. In the stationary bootstrap, one randomly draws a sequence of starting point uniformly from $\{1,2,\ldots,n\}$ and a sequence of Geometric random variable $L_1,L_2,\cdots$. Depending on the values of $L_i$ and the starting value, we draw a consecutive $L_i$-length block and then finally concatenate these blocks together to arrive at a replicated series of length $n$. Then finally we look at the final $m$ for each such series, whereas it is understandable that it was not particularly important to only look at the last $m$ since the starting point is uniform. We conjecture that this ensures consistency of the quantiles despite providing some more dispersion to the average. This extra dispersion arises as this method concatenates blocks containing same elements with a nontrivial probability and thus adding those covariances in the dispersion of the overall average. We postpone a rigorous theoretical justification of this innovative bootstrap technique to a future work since this paper focuses more on the exploration of the performance of LASSO fitted residuals. For the convenience of the readers we summarize the algorithm down below:

\begin{algorithm}[Bootstrap PIs ADJ]\label{meth:ADJ}
\begin{enumerate}
\item[i.] Estimate regression $y_t\sim \tran{\bm{x}}_t$, $t=1\ldots,n$ with LASSO.
\item[ii.] Using stationary bootstrap (See \citet{politis}), replicate residuals $\hat{e}_t=y_t-\hat{y}_t$, $B$ times obtaining $\hat{e}_t^b, t=1,\ldots,n$, $b=1,\ldots,B$.
\item[iii.] Compute $(\bar{e}_{t(m)}^b)=m^{-1} \sum_{i=1}^m e^b_{t-i+1}$, $t=m,\ldots n$ from every replicated series.
\item[iv.] Estimate the $\alpha/2$th and $(1-\alpha/2)$th quantile $\hat{Q}(\alpha/2)$ and $\hat{Q}(1-\alpha/2)$ using Gaussian kernel density estimator from $\bar{e}_{n(m)}^b$, $b=1,\ldots,B$. 
\item[v.] The PI for $\bar{y}_{+1:m}$ is $[L,U]=\bar{\hat{y}}_{n,1:m}+[\hat{Q}(\alpha/2),\hat{Q}(1-\alpha/2)]$, where $\bar{\hat{y}}_{n,1:m}$ is the average of $h$-step-ahead forecasts for $h=1,\ldots,m$.
\end{enumerate}
\end{algorithm}

Our theoretical results are concerned with the consistency of the usual quantiles from the original series. But we conjecture that the stationary bootstrap technique to obtain the replicated series retains the asymptotic dependence structure of the original series and thus the quantiles of the $m$-length average from the original series and that from the final $m$ of the replicated series are close to each other. Additionally, using the Gaussian kernel density to obtain the kernel quantile estimator (See \citet{sm90}) further improves the performance in prediction. These adjustments are supported by the empirical evidence given in \citet{chudyold} in a univariate setup.

Next, we exhibit some theoretical consistency results addressing various different dependency cases, tail-decay and high dimensional regression in the next two sections.

\ignore{ 
 It is important to note that there are other regression estimates in the scenario $p \gg n$ that can work here as well. However, we keep the focus on LASSO only. 
}

\section{Prediction interval for error process} \label{sec:noregress}
\noindent Before moving on to a more general discussion with a large number of covariates compared to sample size, we use this section to discuss a primer without covariates i.e. our response here is just a mean-zero error process.  This section also elaborately describes the model specifications on the error process i.e. whether the process is short/long-range dependent or has light/heavy tails. Under short-range dependence, if window size $m$ is long enough then the dependence of $y_{n+1}+\ldots+y_{n+m}$ on $y_1,y_2,\ldots, y_n$ diminishes and the conditional distribution of $(y_{n+1}+\ldots+y_{n+m})/\sqrt{m}$ given $y_1,\ldots,y_n$ is almost similar to the unconditional distribution and thus one can obtain a simple central limit theorem to quantify the uncertainty surrounding the prediction. Consider $(e_i)$ a mean-zero stationary process and let $S_m=e_1+\ldots+e_m$. \citet{Wu04} proved that for $q>5/2$, the condition 
\begin{eqnarray}\label{eq:sm|f0 cond}
\|\IE(S_m |\mathcal{F}_0)\|_2 = O \left( \frac{\sqrt{m}}{\log ^q m}\right),
\end{eqnarray} gives the a.s. convergence
\begin{equation}\label{eq:sm|f0}
\Delta(\IP(S_m/\sqrt{m} \leq \cdot |\mathcal{F}_0),N(0,\sigma^2))= 0 \text{ a.s.},
\end{equation}
 where $\|\cdot \|_2$ denotes the $L_2$-norm, $\Delta$ denotes the Levy distance, $\mathcal{F}_i$ is the $\sigma-$field $\sigma(\ldots, e_{i-1},e_i)$, $m \to \infty$ and $\sigma^2= \lim_{m \to \infty}\|S_m\|_2^2/m$ is the long-run variance. We start by collecting an asymptotic normality result from \citet{chudyold}. When the error process has possibly long-range dependence such asymptotic normality fails. Keeping the linear structure intact, we provide some empirical consistency result for the quantile-based methods from Section \ref{sec:methods}. These results depict all possible scenarios of heaviness of tail and range of dependence. Finally we extend these results to the more practically applicable non-linear case under the functional dependence framework. Throughout this section, the observed process $e_i$ are either direct linear sum of independent and identically distributed (i.i.d. henceforth) innovations that are unobserved or we assume $e_i$ to be a general non-linear function of these. For the latter, we define some tractable moment-based coupling measure which we call functional dependence measure.

\subsection{Linear error process: Theoretical results}
Assuming linearity of the mean-zero noise process $e_i$ in the following manner 

\begin{eqnarray}\label{eq:ei linear}
e_i= \sum_{j=0}^\infty a_j \epsilon_{i-j},\textrm{ with }\epsilon_i \text{ i.i.d., }\IE(\abs{\epsilon_1}^p)<\infty
\end{eqnarray}
for some $p>2$, it is easy to derive the following conditions on $a_i$ to ensure the convergence in (\ref{eq:sm|f0}). The proof can be found in the appendix of \citet{chudyold}.

\begin{theorem}[Theorem 1 from \citet{chudyold}]\label{th:lighttailed}
Assume the process $e_t$ admits the representation (\ref{eq:ei linear}) where $a_i$ satisfies
\begin{eqnarray}\label{eq:ai formulation}
a_i= O(i^{-\chi}(\log i)^{-A}), \quad \chi>1, A>0,
\end{eqnarray}
\noindent where larger $\chi$ and $A$ means fast decay rate of dependence. Further assume, $A>5/2$ if $1<\chi<3/2$.  Then the sufficient condition (\ref{eq:sm|f0 cond}) implies that the convergence (\ref{eq:sm|f0}) to the normal distribution holds.
\end{theorem}

\begin{remark}
 The central limit theorem described in (\ref{eq:sm|f0}) does not hold if the sequence $a_i$ is not absolutely summable or if the moment assumption in (\ref{eq:ei linear}) is relaxed.
\end{remark}

\noindent Next, as outlined above we proceed for a quantile consistency result for linear processes to validate our proposed methods. We start by formally defining short-range and long-range dependence for linear processes with the representation in (\ref{eq:ei linear}). For short-range dependence we assume

$$\text{(SRDL): }\sum |a_i|<\infty.$$

\noindent For long range dependence, we first revisit the definition of slowly varying function (s.v.f.). A function $g(\cdot)$ is called a s.v.f. if for all $a>0$, $\lim_{x \to \infty}g(ax)/g(x)=1$. Assume 

$$\text{(LRDL($\gamma$))}: l^*(i)=a_i i^{-\gamma} \text{ is a slowly varying function (s.v.f.)} $$

\noindent for some $q<\gamma<1$ where $1/q= \sup \{t:\IE(|\epsilon_j|^t)<\infty\}$. Note that, the definition of $(\text{LRDL}(\cdot))$ also takes into account the possible heavy-tailed distribution of $\epsilon_j$. We also assume $\epsilon_i$ admits a density $f_{\epsilon}$ and 

$$\text{(DENL):} \sup_{x \in \mathbb{R}}(f_{\epsilon}(x)+|f'_{\epsilon}(x)|)<\infty.$$

\noindent For a fixed $0<u<1$, let $\hat{Q}(u)$ and $\tilde{Q}(u)$ denote the $u$-th sample quantile and the actual quantile of $\tilde{S}_i$, $i=m,\ldots,n,$ respectively, where
\begin{equation}\label{eq:yitilde2}
\tilde{S}_i= \frac{\sum_{j=i-m+1}^{i}e_j}{H^*_m}, \quad i=m,m+1,\ldots 
\end{equation}
and
\begin{eqnarray*}
H^*_m= \left. 
  \begin{cases}
    \sqrt{m}, & \text{for Case1},\\
    \inf \{x:\IP(|\epsilon_i|>x) \leq \frac{1}{m} \} & \text{for Case2} ,\\
    m^{3/2-\gamma}l^*(m) & \text{for Case3,} \\
    \inf \{x:\IP(|\epsilon_i|>x)m^{1-\gamma}l^*(m) & \text{for Case4,} \\
  \end{cases}  \nonumber
\right.
\end{eqnarray*}
\noindent where Case 1-2 denotes SRDL holds and  $\IE(\epsilon_j^2)<\infty$ or $\IE(\epsilon_j^2)=\infty$ respectively whereas Case 3-4  stands for $\text{LRDL}(\gamma)$ holds and  $\IE(\epsilon_j^2)<\infty$ or $\IE(\epsilon_j^2)=\infty$ respectively. We collect the following theorem from \citet{zhou10} for the rates of convergence of quantiles depending on the nature of the error process in terms of tail behaviour and dependence:
\begin{theorem}[][Empirical quantile consistency: linear error process] \label{th:light tailed linear}
Assume (DENL) holds. Additionally
\begin{compactitem}[-]
\item Light tailed (SRDL): Suppose (SRDL) holds and $\IE(\epsilon_j^2)<\infty$. If $m^3/n \to 0$, then for any fixed $0<u<1$, 
\begin{equation}
|\hat{Q}(u)-\tilde{Q}(u)|=O_{\IP}(m/\sqrt{n}). 
\end{equation}
\item Light tailed (LRDL): Suppose (LRDL) holds with $\gamma$. If $m^{5/2-\gamma}n^{1/2-\gamma}l^2(n) \to 0$, then for any fixed $0<u<1$, \begin{equation} |\hat{Q}(u)-\tilde{Q}(u)|=O_{\IP}(mn^{1/2-\gamma}|l^*(n)|). \end{equation}
\item Heavy-tailed (SRDL): Suppose (SRDL) holds and $\IE(|\epsilon_j|^{q})<\infty$ for some $1<q<2$. If $m=O(n^{k})$ for some $k<(q-1)/(q+1)$, then for any fixed $0<u<1$, 
\begin{equation}
|\hat{Q}(u)-\tilde{Q}(u)|=O_{\IP}(mn^{\nu})\text{ for all }\nu>1/q-1.
\end{equation}
\item Heavy-tailed (LRDL):  Suppose (LRDL) holds with $\gamma$. If $m=O(n^{k})$ for some $k<(q \gamma -1)/(2q+1-q \gamma)$, then for any fixed $0<u<1$, \begin{equation}|\hat{Q}(u)-\tilde{Q}(u)|=O_{\IP}(mn^{\nu})\text{ for all }\nu>1/q-\gamma.\end{equation}
\end{compactitem}
\end{theorem}
\noindent A technical fact about these results is that their proofs are heavily dependent on the linear structure. We believe that it is an important task to extend these to the more general non-linear scenarios but it comes at the cost of some more abstraction in how these nonlinear processes are posited. We use a very general framework of functional dependence form \cite{wu05} to first describe the class of non-linear process and then provide analogs of central limit theorem and quantile consistency results.

\subsection{Central limit theory for nonlinear dependence}\label{sec:nonlinear}

\ignore{
\todo[inline]{Marek: include GARCH these type of significant and prominent nonlinear processes.}
\todo[inline]{Sayar: I didnt see GARCH being referred to as a nonlinear processes. if the definition of linear is that $y_t$ has causal representation, i.e. is a linear comb. of past shocks, then even if shocks would be of GARCH type, the process $y_t$ would still be linear. Typically, transition autoregression is considered as nonlinear process. I found a reference on nonlinear GARCH, but this was simply a GARCH with transition effect in the equation for variance.}
}

\noindent Economic and financial time series are often subject to structural changes and thus the linear models often fail to capture the diverse range of data generating process. Comparatively possible non-linear class of time-series models is significantly richer in its scope.  Useful nonlinear time-series models, among many, include the regime-switching autoregressive processes, which assume that the series change their dynamics when passing from one regime to another and neural-network models. We provide extension of the asymptotic normality result \eqref{eq:sm|f0} to nonlinear error process. We lay down the definition of non-linearity through the following structure: Let $e_i$ be a stationary process that admits the following representation 
\begin{equation}\label{eq:ei nonlinear}
e_i=H(\mathcal{F}_i)=H(\epsilon_i, \epsilon_{i-1}, \ldots),
\end{equation}
where $H$ is such that $e_i$ is a well-defined random variable, $\epsilon_i, \epsilon_{i-1}, \ldots$ are i.i.d. innovations and $\mathcal{F}_i$ denotes the $\sigma$-field generated by $(\epsilon_{i}, \epsilon_{i-1}, \ldots)$. One can see that it is a vast generalization from the linear structure in $H$. In order to derive a result similar to \eqref{eq:sm|f0} but in the non-linear regime, we define the following functional dependence measure for $e_i$ in \eqref{eq:ei nonlinear}, by which we follow \citet{wu05}'s framework to formulate dependence through coupling: 
\begin{equation}\label{eq:fdm}
\delta_{j,p}=  \| e_i - e_{i, (i-j)} \|_p =  \|H(\mathcal{F}_i)- H(\mathcal{F}_{i, (i-j)} \|_p ,
\end{equation}
 where $\mathcal{F}_{i,k}$ is the coupled version of $\mathcal{F}_i$ with $\epsilon_k$ in $\mathcal{F}_i$ replaced by an i.i.d. copy $\epsilon_k'$,
$\mathcal{F}_{i,k}= (\epsilon_i, \epsilon_{i-1},  \ldots, \epsilon_k', \epsilon_{k-1}, \ldots )$
and $e_{i,(i-j)}= H(\mathcal{F}_{i,(i-j)})$. Clearly, $\mathcal{F}_{i,k}= \mathcal{F}_i$ is $k>i$. As \citet{wu05} suggests, $\|H(\mathcal{F}_i)- H(\mathcal{F}_{i, (i-j)}) \| _p$ measures the dependence of $e_i$ on $\epsilon_{i-j}$. This dependence measure can be seen as an input-output system and a natural analogue of the linear coefficient at a certain lag. It facilitates mild moment conditions on the dependence of the process which are easily verifiable compared to the more popular strong mixing conditions. Define the cumulative dependence measure  
\begin{equation}\label{eq:cumfdm}
\Theta_{j,p}=\displaystyle\sum_{i=j}^{\infty}\delta_{i,p},
\end{equation}
which can be thought as cumulative dependence of $(e_j)_{j \geq k}$ on  $\epsilon_k$. 
 
\begin{theorem}\label{th:quenched clt}
Assume $e_i$ admits the representation in (\ref{eq:ei nonlinear}). Also assume under the functional dependence formulation in \ref{eq:fdm}, the following rate holds for the cumulative dependence $\Theta_{j,p}$: 
\begin{equation}\label{eq:thetaip condition}
\Theta_{j,p} =O(j^{-\chi} (\log j)^{-A}) \text{ where }
\begin{cases}
    A>0 \text{ for } 1<\chi<3/2,\\
    A>5/2 \text{ for } \chi \geq 3/2,
  \end{cases}
\end{equation}
then the convergence in (\ref{eq:sm|f0}) holds. 
\end{theorem}

Apart from the asymptotic normality result from Theorem \ref{th:quenched clt}, we will next show that the estimated quantiles from the moving blocks are consistent. However, the same notion of non-linear dependence falls short of establishing quantile consistency due to a technical reason and thus we resort to first define some form of predictive dependence.

\ignore{
\begin{remark}
For a linear process $e_i$ with $$e_i=\sum_{j=0}^{\infty} a_j \epsilon_{i-j} \quad e_i \text{ i.i.d. with } \IE|\epsilon|_1^p)<\infty, \sum_{j=0}^{\infty} |a_j|<\infty$$ it is easy to obtain the functional dependence measure $\delta_{j,p}=a_j$ and thus the condition (\ref{eq:thetaip condition}) can be simplified in terms of the coefficients $a_j$ 
\end{remark}
}

\subsection{Quantile consistency for nonlinear process: predictive dependence}
\noindent For a general nonlinear process with possibly heavy tails and long-range dependence, the central limit theorem fails. We present one of the main results in this paper by showing the empirical quantile consistency that validates PIs of the form \eqref{eq:EQpis}. Towards that, we need to  control the latent dependence of $e_i$ on $\epsilon_{i-j}$ keeping in mind a representation like (\ref{eq:ei nonlinear}). Therefore, we introduce the predictive density-based dependence measure.  Let $\mathcal{F}'_k = (\epsilon_k, \ldots , \epsilon_{1}, \epsilon_0', \epsilon_{-1}, \ldots ,),$ be the coupled shift process derived from $\mathcal{F}_k$ by substitution of $\epsilon_0$ by its i.i.d. copy $\epsilon_0'$. Let $F_1(u|\mathcal{F}_k) = P\{H(\mathcal{F}_{k+1}) \leq u|\mathcal{F}_k\}$ be the one-step ahead predictive or conditional distribution function and $f_1(u|\mathcal{F}_k) = d F_1(u|\mathcal{F}_k)/ du,$  be the corresponding conditional density.  We define the predictive dependence measure
\begin{equation}\label{eq:pred density}
\psi_{k,q} = \sup_{u \in \mathbb{R}} \| f_1(u|\mathcal{F}_k) - f_1(u|\mathcal{F}'_k) \|_q.
\end{equation}
The quantity $\psi_{k,q}$ (\ref{eq:pred density}) measures the contribution (read change in $q$-th norm) of $\epsilon_0$, the innovation at step 0, on the conditional or predictive density at step $k$. We shall make the
following assumptions:
\begin{enumerate}[i.]
\item For short-range dependence: $\Psi_{0,2}<\infty$	where $\Psi_{m,q}=\sum_{k=m}^{\infty}\psi_{k,q}$; for long-range dependence: $\Psi_{0,2}$ can possibly be infinite;
\item  (DEN) There exists a constant $c_0 <\infty$ such that almost surely,
$$ \sup_{u \in \mathbb{R}} \{f_1(u|\mathcal{F}_0) +|d f_1(u|\mathcal{F}_0)/ d u|\} \leq c_0.$$
\end{enumerate}
The (DEN) implies that the marginal density $f(u) = \IE{f_1(u|\mathcal{F}_0)} \leq 
c_0.$ Recall the sufficient conditions for the linear cases in \citet{zhou10} were based on the coefficients of the linear process. Here, the conditions for both short-range and long-range dependent errors here were transferred onto predictive dependence measure. We assume:
\begin{eqnarray} \label{eq:conditions2}
(\text{SRD})&:& \sum_{j=0}^{\infty}|\psi_{j,q}|<\infty,\\ \nonumber \label{eq:lrd2}(\text{LRD($\gamma$)})&:& \psi_{j,q}=j^{-\gamma} l(j),  q< \gamma<1, l(\cdot) \text{ is a slowly varying function (s.v.f.) }.
\end{eqnarray}
where $1/q= \sup \{t:\IE(|\epsilon_j|^t)<\infty\}$. For a fixed $0<u<1$, let $\hat{Q}(u)$ and $\tilde{Q}(u)$ denote the $u$-th sample quantile and actual quantile of $\tilde{S}_i$; $i=m,\ldots,n,$ where
\begin{equation}\label{eq:yitilde}
\tilde{S}_i= \frac{\sum_{j=i-m+1}^{i}e_j}{H_m}, \quad i=m,m+1,\ldots 
\end{equation}
and
\begin{eqnarray}\label{eq:Hm}
&& \\
&&\hspace{-0.27 in}H_m=\left. 
  \begin{cases}
    \sqrt{m}, & \text{for Case1},\\
    \inf \{x:\IP(|\epsilon_i|>x) \leq \frac{1}{m} \} & \text{for Case2} ,\\
    m^{3/2-\gamma}l(m) & \text{for Case3,} \\
    \inf \{x:\IP(|\epsilon_i|>x)m^{1-\gamma}l(m) & \text{for Case4,} \\
  \end{cases}  \nonumber
\right.
\end{eqnarray}
\noindent Here, Case 1-2 denotes $SRD$ holds and  $\IE(\epsilon_j^2)<\infty$ or $\IE(\epsilon_j^2)=\infty$ respectively whereas Case 3-4  stands for $LRD(\gamma)$ holds and  $\IE(\epsilon_j^2)<\infty$ or $\IE(\epsilon_j^2)=\infty$ respectively.

\noindent Then we have following rates of convergence of quantiles depending on the nature of the error process in terms of tail behaviour and dependence:
\begin{theorem}[Empirical quantile consistency: nonlinear error process] \label{th:light tailed nonlinear}
\begin{compactitem}[-]
\item Light tailed (SRD): Suppose (DEN) and (SRD) hold and $\IE(\epsilon_j^2)<\infty$. If $m^3/n \to 0$, then for any fixed $0<u<1$, 
\begin{equation}
|\hat{Q}(u)-\tilde{Q}(u)|=O_{\IP}(m/\sqrt{n}). 
\end{equation}
\item Light tailed (LRD): Suppose (LRD) and (DEN) hold with $\gamma$ and $l(\cdot)$ in (\ref{eq:conditions2}). If $m^{5/2-\gamma}n^{1/2-\gamma}l^2(n) \to 0$, then for any fixed $0<u<1$, \begin{equation} |\hat{Q}(u)-\tilde{Q}(u)|=O_{\IP}(mn^{1/2-\gamma}|l(n)|). \end{equation}
\item Heavy-tailed (SRD): Suppose (DEN) and (SRD) hold and $\IE(|\epsilon_j|^{q})<\infty$ for some $1<q<2$. If $m=O(n^{k})$ for some $k<(q-1)/(q+1)$, then for any fixed $0<u<1$, 
\begin{equation}
|\hat{Q}(u)-\tilde{Q}(u)|=O_{\IP}(mn^{\nu})\text{ for all }\nu>1/q-1.
\end{equation}
\item Heavy-tailed (LRD):  Suppose (LRD) hold with $\gamma$ and $l(\cdot)$ in (\ref{eq:conditions2}). If $m=O(n^{k})$ for some $k<(q\gamma -1)/(2q+1-q \gamma)$, then for any fixed $0<u<1$, \begin{equation}|\hat{Q}(u)-\tilde{Q}(u)|=O_{\IP}(mn^{\nu})\text{ for all }\nu>1/q-\gamma.\end{equation}
\end{compactitem}
\end{theorem}
\noindent Note that the results presented above talks about the scenario without any predictor. For the low dimensional regression case, the proofs follow exactly similar to how \citep{zhou10} proceeded by showing $$\sup_{i \leq n}|\hat{e}_i-e_i|=O_{\IP}(\Pi(n)) $$for some suitably chosen small $\Pi(n)$. We skip the details here and directly move on to the scenario where the number of predictors far outgrows the number of past time-points available.


\section{High dimensional regression regime}\label{sec:high dimension}
In this section, we focus on results concerning theoretical guarantees of the prediction intervals for high-dimensional regression where the error process is possibly nonlinear and shows temporal dependence. Note that all the results from this section can be easily extended to situations where the error process has exponentially decaying tails such as sub-exponential and sub-gaussian, but we restrict ourselves to conditions of only finitely many moments which is a significantly weaker assumption than that is present in the literature. 

The results for this subsection heavily depend on optimal concentration inequalities for $S_{n,b}=\sum_{i=1}^nb_ie_i$ that have not been established for possibly nonlinear dependence before this. Since this is just a proof technique tool we postpone the concentration results here. However, as one of our final aim is to point out the price of dependence in a wide range of general scenarios, it is important to define dependence-adjusted norm to be able to state the concentration and consistency results. Recall (\ref{eq:cumfdm}) and assume the short-range dependence, $\Theta_{0,q}<\infty$ holds with $q$ being less or more than 2 depending on the tail-behavior of the error process. Further, we define dependence adjusted norm, for $\alpha>0$,
\begin{eqnarray}\label{eq:dep adjusted}
\|e. \|_{q, \alpha}= \sup_{t \geq 0}(t+1)^{\alpha}\sum_{i=t}^{\infty} \delta_{i,q}.
\end{eqnarray}
It is easy to note that, finiteness of the dependence-adjusted measure is a stronger ask than the finiteness of $\Theta_{0,q}$. Next, we show that for short-range dependent nonlinear error processes the error bounds obtained in Theorem \ref{th:light tailed nonlinear} remain intact under a proper choice of the sparsity condition.

\ignore{
Note that the theorem uses the following crucial lemma to use Nagaev inequality from Theorem \ref{th:nagaev linear light} and  \ref{th:nagaev nonlinear light} to $V_j$.
\begin{result}\label{eq:bickel lemma}
\citep{bickel09}
Let $\lambda=2r$ in (\ref{eq:lasso objective}). Also  assume,\\ $ r= \max ( A (n^{-1}\log p )^{1/2} \| e_.\|_{2, \alpha}, B \|e_{.} \|_{q,\alpha} \|X\|_q/n)$. On the event $\mathcal{A}= \bigcup_{j=1}^p \{2|V_j| \le r \} , \text{ where } V_j =\frac{1}{n}\sum_{i=1}^n e_ix_{ij},$ we have, $r\|\hat{\bm \beta}-\bm \beta\|_1+\|X(\hat{\bm \beta}-\bm \beta)\|_2^2/n \leq 4r \|\hat{\bm \beta}_J- \bm \beta_J\|_1 \leq 4r\sqrt{s} \|\hat{\bm \beta}_J- \bm \beta_J\|_2.$
\end{result}

Next, we move on to discuss the quantile consistency for the Lasso fitted residuals as described in Section \ref{sec:methods}. 
}

\subsection{Lasso with fixed design}\label{ssc:lasso fixed}
\noindent For the model in (\ref{eq:regression}), we first assume that the $\bm x_i$'s are fixed and the future $\bm x_i$'s are known. Under this setting, we next show the quantile consistency for the nonlinear process. Note that a very similar result can be shown for the linear process with the error process admitting a simpler representation (\ref{eq:ei linear}), however we skip writing that as a separate theorem here to avoid repetitiveness. 

\ignore{
\begin{theorem}\label{th:covariate}
(Empirical quantile consistency for LASSO-linear)
Let $\bar{Q}_n(u)$ be the $u$-th empirical quantile of  $(\tilde{\hat{S}}_i)_m^n$. Denote the number of non-zero elements of $\bm \beta$ by $s$ and assume that
\begin{eqnarray}\label{eq:sparsity condition}
\text{ (for light tails)}\quad s &\ll& \min\left(\sqrt{n/\log p}, \frac{n^{1-\max\{0,1/2-1/q-\alpha\}}}{\norm{X}_q}\right), \\
\text{ (for heavy tails)}\quad s &\ll& \frac{n^{\frac{2}{\alpha}+\frac{\alpha-1}{\alpha+1}}L_1(n)}{\norm{X}_q},  \end{eqnarray}
then conclusions of Theorem \ref{th:light tailed nonlinear} with $Q_n(u)$ replaced by $\bar{Q}_n(u)$.
\end{theorem}
}



\begin{theorem}\label{th:covariate2}
(Empirical quantile consistency for LASSO-nonlinear)


\noindent Assume the covariates are so scaled such that $\|X\|_2=(np)^{1/2}$. Denote $\lambda=2r$ in the criterion function (\ref{eq:lasso objective}) where
\begin{eqnarray}\label{eq:r}
r=\max\{A \sqrt{n^{-1}\log p}\|e_. \|_{2,\alpha}, B \|e_.\|_{q,\alpha}\|X\|_qn^{-1+\min\{0,1/2-1/q-\alpha\}} \}. \nonumber
\end{eqnarray}
We assume that the restricted eigenvalue assumption RE($s, \kappa$) in \cite{bickel09} holds with constant $\kappa = \kappa(s, 3)$, where $s$ is the number of non-zero entries in true parameter vector $\bm \beta$ and 
\begin{eqnarray}\label{eq:kappa}
\kappa(s,c)=\min_{J \subset \{1,\cdots,p\},|J|\leq s,} \min_{|u_{J^c}|_1 \leq c |u_{J}|_1}\frac{\|Xu\|_2}{\sqrt{n}\|u_J\|_2}.
\end{eqnarray}

\noindent Here $u_J$ stands for modified  $u$ by setting its
elements outside $J$ to zero. Let $\bar{Q}_n(u)$ be the $u$-th empirical quantile of  $(\tilde{\hat{S}}_i)_m^n$. Assume that (SRD) holds and for $r$ defined in (\ref{eq:r}),
\begin{eqnarray}\label{eq:sparsity condition}
\text{ (for }q \geq 2)\quad s &=& o\left(\frac{m}{r^2n}\right), \nonumber\\
\text{ (for }1 <q \le 2), \quad s &=& o \left(\frac{H_m^2|l(n)|^2}{r^2n^{2 \gamma-1}} \right), 
\end{eqnarray}
where $\gamma$ and $l(\cdot)$ are defined in (\ref{eq:conditions2}), $H_m$ in (\ref{eq:Hm}) and $\alpha$ in the definition of $r$ is in the context of the dependence adjusted norm defined in (\ref{eq:dep adjusted}), then the (SRD) specific conclusions of Theorem \ref{th:light tailed nonlinear} hold with $Q_n(u)$ replaced by $\bar{Q}_n(u)$.
\end{theorem}
\noindent One can note the $\sqrt{\log p/n}$ term we have in our definitions for $r=\lambda/2$. This allows us to capture the ultra-high dimensional scenario where $\log p =o(n)$, the usual benchmark in the high-dimensional literature. The additional terms involving $\|e_{.}\|_{.,\alpha}$ are due to the dependence present in the error process. The sparsity condition for the light tail case, i.e. $q \geq 2$,  in (\ref{eq:sparsity condition}) in the view of the choice of $r$ in (\ref{eq:r}) can be written as: For $q\geq2$ 

\begin{eqnarray*}
s \ll \min\left(\frac{n^{4/3}}{\log p \|e_.\|^2_{2,\alpha}}, \frac{n^{7/3-2\max\{0,1/2-1/q-\alpha\}}}{|X|_q^2\|e_.\|^2_{q,\alpha}}\right)
\end{eqnarray*}

\noindent for the choice of $m=o(n^{1/3})$. Thus we can allow $p$ to grow a bit faster than the usual ultra-high dimensional benchmark $e^{O(n)}$ rate. This is an interesting result and our conjecture is that this added advantage is due to considering future aggregation instead of just $k$-step ahead forecast for a fixed $k$. In other words, if prediction horizon $m$ is allowed to grow to $\infty$, the $m$-length average of residuals can automatically provide some concentration. Thus it can allow for scenarios where estimation of $\bm \beta$ is not very precise. We believe that this is an interesting exploration of the relaxation of the sparsity condition compared to the usual LASSO literature.

For the special case of the linear process (See \ref{eq:ei linear}), the conditions in (\ref{eq:sparsity condition}) will remain identical and one can provide some more specifications in the definition of $r$ in (\ref{eq:r}) using the linear coefficients $a_i$ from (\ref{eq:ei linear}) in the view of the Nagaev-type concentration inequalities derived in Result \ref{th:nagaev linear light}. Moreover, for the linear process, one can also state the corresponding results for the (LRDL) case, however since the condition on sparsity is unaffected by this nature of dependence, we do not state them separately here.

\subsection{Lasso with stochastic design: covariate prediction issue}\label{ssc:lasso stochastic}
\noindent For the high-dimensional regression scenario, it is very natural to ask whether one could relax the fixed predictor settings to random design as it is practically impossible to find many perfectly predictable covariates. We show that under a mild condition on the possibly stochastic covariate process, the prediction intervals based on LASSO-optimized residuals also work in a stochastic design set-up. This is a particularly interesting extension from \citet{zhou10} since in the high-dimensional regime it is impractical to assume an exponentially growing number of covariates to be perfectly predictable for its future values.  Moreover, apart from allowing a random design, we also allow the covariates to be temporally dependent. For this subsection, we restrict ourselves to only nonlinear processes for the sake of clarity. We assume

\begin{eqnarray}\label{eq:xi}
\bm x_t=\textbf{G}_x(\epsilon_t^x,\epsilon_{t-1}^x,\ldots)=(g_1(\epsilon_t^x,\epsilon_{t-1}^x,\ldots), g_2(), \ldots, \tran{g_p(\epsilon_t^x,\epsilon_{t-1}^x,\ldots))}
\end{eqnarray}
where $\epsilon_i^x$ are i.i.d., $\textbf{G}_x$ is a measurable function and $\bm x_t$ is $p \times 1$ high-dimensional mean-zero stationary process. Let $\{\epsilon_{i}^{x'}\}$ be an i.i.d. copy of $\{\epsilon_i^x\}$. Due to the high dimension of $\bm x_t$ and the dependence as a stochastic time-series, we opt for a dependence-adjusted uniform (over co-ordinates) functional dependence measure. For a fixed co-ordinate $1 \leq j \leq p$ and $q \geq 1$, define

\begin{eqnarray}\label{eq:xi fdm}
\delta_{k,q,j}^{\bm{x}}=\|g_j((\epsilon_i^x,\epsilon_{i-1}^x,\ldots,\epsilon_{i-k}^{x},\ldots))-g_j((\epsilon_i^x,\epsilon_{i-1}^x,\ldots,\epsilon_{i-k}^{x'},\ldots))\|_q.
\end{eqnarray}
Use this to define the dependence-adjusted norm as 
\begin{eqnarray}\label{eq:Phialpha}
\Phi_{q,\alpha}=\sup_{1 \leq j \leq p}\|X_{.j}\|_{q,\alpha}\text{ where }\|X_{.j}\|_{q,\alpha}=\sup_{m \geq 0} (m+1)^{\alpha} \sum_{t=m}^{\infty}\delta_{t,q,j}.
\end{eqnarray}
The quantity $\Phi_{q,\alpha}$ provides a concise and natural measure of dependence which can effectively account for high dimensionality and temporal dependence. Let also the error process $(e_i)$ admit the following representation 

\begin{eqnarray}\label{eq:ei}
e_i=G_e(\epsilon_i^e,\epsilon_{i-1}^e,\ldots),
\end{eqnarray}
where $\epsilon_i^x$ are i.i.d. and $G_e$ is a measurable function. One can then define the cumulative dependence using this functional dependence measure. However, thanks to the structure of the regression problem we can skip defining that for the $e_i$ process. Instead, for the quantile consistency in the case of stochastic design, we will need a notion of functional dependence on the cross-product process $x_{.j}e_.$ as follows 

\begin{eqnarray}\label{eq:fdm ex}
\delta_{k,q}^{xe}=\max_{j \le p}\|x_{lj}e_l-x^*_{k,lj}e_{k,l}^*\|_{q}, \end{eqnarray}
where $e_{k,l}^*=G_e(\epsilon_l^e,\epsilon_{i-1}^e,\ldots,\epsilon_{i-k}^e{'},\ldots)$ and $\{\epsilon_{i}^e\}'$ is an i.i.d. copy of $\{\epsilon_i^e\}$. Using (\ref{eq:fdm ex}), we define the dependence adjusted norm as (\ref{eq:dep adjusted}) for the $x_{.j}e_.$ process uniformly over $1 \le j \le p$ as follows:

\begin{eqnarray}\label{eq:dep adjusted ex}
\max_{j \leq p} \|x_{.j}e. \|_{q, \alpha}= \sup_{t \geq 0}(t+1)^{\alpha}\sum_{i=t}^{\infty} \delta_{i,q}^{xe}.
\end{eqnarray}

\ignore{
\begin{theorem}[Nagaev inequality for nonlinear processes]
heavy-tailed, short-range]\label{th:nagaev nonlinear heavy}
Assume that $\|e.\|_{q,\alpha} <\infty$ where $q>2$ and $\alpha>0$ and $\sum_{i=1}^n b_i^2=n.$ Let $r_n=1$(resp. $(\log n)^{1+2q}$ or $n^{q/2-1-\alpha q}$ ) if $\alpha>1/2-1/q$ (resp. $\alpha=0$ or $\alpha< 1/2-1/q$). Then for all $x>0$,
\begin{equation}\label{eq:nonlinear heavy}
\IP|S_{n,b}| \geq x) \leq C_1 \frac{r_n}{(\sum_{j}|b_j|)^q\|e.\|_{q,\alpha}^q}{x^q}+C_2 \exp \left ( -\frac{C_3x^2}{n \|e.\|^2_{2,\alpha}}\right),
\end{equation} 
 for constants $C_1,C_2,C_3$ that depend on only $q$ and $\alpha$.
\end{theorem}
}

\noindent With this background on the stochastic covariate and the error process we now state the quantile consistency result below. 
\begin{theorem}\label{th:lasso stochastic}
(Empirical quantile consistency for LASSO-stochastic)

\noindent Assume $\Phi_{q,\alpha} <\infty$ and $ \max_{j\le p}\|x_{.j}e.\|_{q,\alpha} <\infty$ for some $\alpha>0$. Let $\bar{Q}_n(u)$ be the $u$-th empirical quantile of  $(\tilde{\hat{S}}_i)_m^n$. We assume that the restricted eigenvalue assumption $\text{RE}_{stoch}(s, \kappa)$ holds with constant $\kappa_{stoch}:= \kappa(s, 3)$, where $s$ is the number of non-zero entries in true parameter vector $\bm \beta$ and 
\begin{eqnarray}\label{eq:kappastoch}
\kappa(s,c)=\sqrt{\min_{J \subset \{1,\cdots,p\},|J|\leq s,} \min_{|u_{J^c}|_1 \leq c |u_{J}|_1}\frac{\tran{u} \IE(\bm x_t \tran{\bm x_t}) u}{\|u_J\|_2^2}}.
\end{eqnarray}

\noindent Assume that the sparsity conditions in (\ref{eq:sparsity condition}) hold with the choice of $\lambda=2r$ where

\begin{eqnarray}\label{eq:new r}
r&=&\max\{A \sqrt{n^{-1}\log p}\max_{j \leq p}\|x_{.j}e_. \|_{2,\alpha}, B \max_{j \leq p}\|x_{.j}e_.\|_{q,\alpha}n^{-1+\min\{0,1/2-1/q-\alpha\}} \}.\nonumber
\end{eqnarray}
Additionally $s$ satisfies $s\sqrt{\log p}/\sqrt{n} \to 0$. Then the SRD specific conclusions of Theorem \ref{th:light tailed nonlinear} hold with $Q_n(u)$ replaced by $\bar{Q}_n(u)$.
\end{theorem}

\begin{remark}
Note that, the uniform functional dependence measure on the cross-product space $x_{.j}e_.$ can often be simplified using H\"older inequalities and the usual triangle inequality technique. For some examples and calculations of the functional dependence measure for the nonlinear covariate processes, see \citet{wuwu16}.
\end{remark}

The prediction intervals based along the line of (\ref{eq:pi}) would need the future values of $\bm{x}_i$ which we do not observe. One possible solution to this is to fit a vector-autoregressive (VAR) model with appropriate lags and then estimate the $k$-step ahead predictions for $ 1 \leq k \leq m$ using the estimated matrix coefficients of the VAR process. It is also possible to lay down assumptions on the $\bm{x}_i$ and $e_i$ process and handle this in a much more rigorous way. But since our focus is on the relatively easier but asymptotically valid methods of estimating the prediction intervals based on quantiles, we omit that discussion. Instead, for practical implementation, (cf. our data analysis from Section \ref{sec:real data analysis} where, for covariates, we used wind and temperature data that are stochastic in nature) we just use the past year-over-year mean as to substitute for $\tran{\bm{x}_i}\bm \beta $ for $i =n+1 ,\ldots,n+m$. Since $m$ is also growing, it is natural that the values of aggregated $\tran{\bm{x}_i}\bm \beta$ would concentrate around aggregated past year-over-year mean. We plan to discuss the issue of simultaneously estimating the future $\bm{x}_i$'s in a future work.

\section{Simulation}\label{sec:2ndPOOS}\label{sec:simu}
\noindent In this section, we compare the predictive performance of methods discussed above using OLS, LAD and LASSO estimator in both low-dimensional and high-dimensional setup. In the low-dimension setup the comparison between CLT based methods (Quenched CLT as described in Section \ref{sec:methods}) and QTL- Quantile based methods will be evaluated. 

\subsection{Simulation set-up}\label{sec:simuSetup}
\noindent The focus here is on evaluation of PIs discussed in the previous section based on their coverage probability. We start by generating the error process $(e_t)$ as:

\begin{enumerate}
\item[(a)]  $e_i=\phi_1 e_{i-1}+\sigma\epsilon_i$, 
\item[(b)]  $e_i=\sigma\sum_{j=0}^\infty(j+1)^{\gamma}\epsilon_{i-j}$, 
\item[(c)]  $e_i=\phi_1e_{i-1}+G(e_{i-1};\delta,T)(\phi_2e_{i-1})+\sigma\epsilon_i$,
\end{enumerate}

\noindent with $\epsilon_i$ i.i.d. from an $\alpha^*$-stable distribution. The  heavy-tails index $\alpha^*=1.5$,  autocovariance decay parameter $\gamma=-0.8$, speed-of-transition parameter $\delta=0.05$, autoregressive coefficients $\phi_1=0.6$ and $\phi_2=-0.3$, the noise standard deviation $\sigma=54.1$, and threshold $T=0$, were all selected based on the autoregressive models fitted to the electricity prices used later in the empirical part. The logistic transition function is given by $G(e_{i-1};\delta,T)=(1+\exp(-\delta(e_{i-1}-T)))^{-1}$. These three specifications represent 

\begin{enumerate}[(a)]
    \item a heavy-tail and short-memory error-process,
    \item a heavy-tail and long-memory error-process, and
    \item a nonlinear error-process know as the logistic smooth transition autoregression (LSTAR) with heavy-tailed innovations.
\end{enumerate}
respectively. Eventually, we add a large number of exogenous covariates to the error process, obtaining $y_i= \tran{\bm{x}}_i\bm \beta +e_i, i=1,\ldots,n+m.$
We compute our PIs based on $(y_1,\bm x_1)\ldots,(y_n,\bm x_n)$ and evaluate them on $\bar{y}_{+1:m}=1/m\sum_{i=1}^my_{n+i}$. Note that we predict the averages instead of sums. This is motivated by easier comparison of predictive performance across different forecast horizons, and also turns out as more appropriate in the following empirical part. 

Regarding covariates, we consider two scenarios (i) $p<n$ and (ii) $p>n$. In scenario (i) (See Table \ref{tab:SIM_50_LD_uniform }), we compare PIs based on OLS, LAD, and LASSO estimators. We set $n=8736$ ($\approx$ 1 year of hourly data), $m=168,336,504,672$ (1,2,3,4 weeks of hourly data), and $p=319$ (151 weather variables and 168 periodic variables), similarly to our empirical application (except that the horizon there spans up to 17 weeks) described in Section \ref{sec:real data analysis}. In scenario (ii) (See Table \ref{tab:SIM_50_HD_uniform }), we only use the LASSO as the other two estimators are not uniquely identified.  We set\footnote{We reduce $n$ for computational convenience so that the $p>n$ does not have to be very large.} $n=336$ (2 weeks of hourly data), $m=24,48,72,96$ (1,2,3,4 days of hourly data) and $p=487$ (151 weather variables and 336 periodic variables). 

The elements of $\bm \beta\in\IR^p$ are i.i.d. from the uniform distribution\footnote{Previous version of this paper also contained results based on Cauchy distribution. Leading to the same general conclusions, we have omitted them for the sake of brevity.} $U[-1,1]$. Moreover, as properties of the LASSO estimator depend on the sparsity of $\bm \beta$, we assume $s=(1-\|\bm \beta\|_0/p)=50\%$. Supplementary online material contains results for sparsity $s=90\%$ and $20\%$, i.e., for high and low sparsity set-up. Throughout the experiment, we keep the (sparse) $\bm \beta$ fixed for all $1000$ repetitions. We compute PI for nominal coverage $(1-\alpha) = 60\%, 80\%, 90\%, 95\%$ and compare the 
quenched CLT method and QTL method proposed in Section \ref{sec:methods}, based on their coverage probabilities (CP hereafter with the plural being CPs)
\[
    (\widehat{1-\alpha}) = \frac{1}{1000} \sum_{j=1}^{1000} \mathbb{I}\left([L,U]_{j,\hat{\bm \beta}}\ni\bar{y}_{j,+1:m}\right),
\]
and based on the average Winkler loss \citep{winkler72}
\[
    \mathcal{L} = |[L,U]_{j,\hat{\bm \beta}}| + \frac{2}{\alpha} \inf_{z\in [L,U]_{j,\hat{\bm \beta}}} |\bar{y}_{j,+1:m}-z|, \quad j=1,\ldots, 1000,
\]
where $\mathbb{I}$ for the $j$-th trial is 1 when $\bar{y}_{j,+1:m}$ is covered by the interval $[L,U]_{j,\hat{\bm \beta}}$ and 0 otherwise. The Winkler loss is a commonly used quantile-type loss, which penalizes the width of the PI and the size of misses, thus being appropriate for comparing PIs based on known quantiles \citep[see][for a recent discussion]{askanazi18}. Notably, the higher is the nominal coverage $1-\alpha$ the more weight is assigned to the misses via the factor $2/\alpha$.

\subsection{Simulation results}\label{sec:simuResults}
Two general conclusions may be drawn from the experiment, independent of whether the dimensionality of covariates is high or low:
\begin{itemize}
  \item For all methods, their CP is always below the nominal coverage for which they were calibrated.
  \item Moreover, CP sinks with the growing forecast horizon. By contrast, the Winkler loss does not increase monotonically with the growing horizon, which indicates stability in terms of the trade-off between CP and \enquote{sharpness} across these forecast horizons.
\end{itemize}
Further results require the distinction of the low-dimensional and high-dimensional scenarios. In the low dimensional scenario (see Table \ref{tab:SIM_50_LD_uniform }), in general, the following holds:
\begin{itemize}
    \item PIs have the lowest CP when the underlying process has a long memory.
    \item QTL PIs dominate the quenched CLT across all series both in terms of CP and  Winkler loss. Hence, CLT cannot compensate out the CP by \enquote{sharpness}. The CP can become lower than half the nominal coverage for the long horizon. 
    \item QTL PIs perform best when based on LASSO and LAD estimators. While LASSO QTL dominates in terms of CP, the LAD CTL has better \enquote{sharpness}, leading to a slight preference of LAD over LASSO based on Winkler loss. The exception from this rule is when series exhibit long memory. Allowing higher sparsity in $\bm{\beta}$ would make LASSO the Winkler loss winner, while lower sparsity would empower the LAD. 
\end{itemize}
For the high-dimensional scenario, some of the previous statements do not hold. Similar to the design in \citet{chudyold}, the horizon/sample ratios become very high, i.e., $m/n>1/4$. Moreover, we must face the curse of dimensionality for which we use LASSO\footnote{Details concerning the selection of tuning parameter $\lambda$ for LASSO can be found in Appendix C}. Still, this set-up has a largely negative impact on QTL's performance (see Table \ref{tab:SIM_50_HD_uniform }), and therefore we exploit a data-driven adjustment based on replication of the residual $\hat{e}_i=y_i-\hat{y}_i$ using stationary bootstrap. We denote the adjusted QTL  by ADJ. Details concerning implementation are in  Section \ref{sec:real data analysis}, where ADJ is used under similar a set-up as here. The simulations provide us with the following results:
\begin{itemize}
    \item QTL PIs  do not dominate the quenched CLT across all series and horizons. In fact, for the shortest horizon, CLT wins in terms of CP of Winkler loss (at least for the large nominal coverage). Moreover, for long-memory series, CLT wins across all horizons.
    \item In general, the CP of QTL is worse than in the previous set-up, especially if the forecast horizon is long. 
    \item However, the bootstrap adjustment introduced in Section \ref{ssc:bootstrap} leads to major improvement across all series and all horizons (except the shortest one). In terms of Winkler loss, the improvement is most visible in the case of non-linear series. 
\end{itemize}
Additionally, we may wonder what the impact of the sparsity and the generating distribution of $\bm{\beta}$ is. In general, the CPs are slightly higher when $\bm{\beta}$ is very sparse. In turn, $\bm{\beta}$ drawn from the Cauchy distribution leads to slightly smaller CPss. Still, both alternative set-ups lead to the same general conclusions with differences between CPss of identical methods (for identical series and horizons) within the range of two percentage points. 
\clearpage
\renewcommand{\thesubtable}{\roman{subtable}}
\begin{sidewaystable}
\vspace{1cm}
\begin{subtable}{1\textwidth}
\scalebox{0.6}{
 \input{Tables/TAB_SIM_50_LD_uniform.tex}
 }
\caption{\footnotesize{Scenario $n>p$. QTL implemented using each of the estimators OLS, LAD or LASSO and CLT using LASSO only.\vspace{0.1cm}}}\label{tab:SIM_50_LD_uniform }
\end{subtable}
\begin{subtable}{1\textwidth}
\scalebox{0.6}{
 \input{Tables/TAB_SIM_50_HD_uniform.tex}
 }
\caption{ \footnotesize{Scenario $p>n$. QTL and CLT as above. ADJ is a bootstrap version QTL for better performance under short-sample.}}\label{tab:SIM_50_HD_uniform }
\end{subtable}
\caption{\footnotesize{Simulated out-of-sample forecasting experiment. The reported values are coverage probabilities, i.e., relative (\%) counts of out-of-sample values covered in 1000 trials (left part) and average Winkler loss values (right part). The nominal coverage (the first column) ranks from is $95\%$ to $60\%$. Simulated error processes heavy tails and either short memory, long memory or are nonlinear. The elements of regression coefficient $\bm \beta$ are drawn independently from uniform distribution $U[-1,1]$. The sparsity of $\bm \beta$'s is fixed to 50\%. For convenience, the best value for each nominal coverage and horizon is marked in bold.}  }\label{tab:SIM_50_uniform }
\end{sidewaystable}
\clearpage

\section{Real data: EPEX Spot electricity prices}\label{sec:real data analysis}
\noindent Next we compare and contrast forecasts obtained by our methods with the existing ones through POOS in a real-life data. Following is a list of competing methods we will explore:
\begin{itemize}
\itemsep0em
\item[(ADJ)] Adjusted QTL-LASSO method described below in the Methods subsection \ref{ssc:bootstrap}.
\item[(RBS)] Robust Bayes (\citet{mw16}),
\item[(ARX)] Bootstrap path simulation from ARMAX models,
\item[(ETS)] Exponential smoothing state-space model \citep{hkos08},
\item[(NAR)] Neural network autoregression \citep[][sec. 9.3]{ ha13}.
\end{itemize}

\noindent We first give a quick overview of the last three methods above so that they are comparable to the ADJ as specified in Section \ref{ssc:bootstrap}.

\noindent \textit{Robust Bayes PIs RBS}:\\
For this sophisticated univariate approach, we focus on intuition and refer to the supplementary Appendix of \citet{mw16} for more details about the implementation. The \textit{robust Bayes} PIs are specifically designed for long-horizon predictions, e.g., when $m/n\approx1/2$. First, the high-frequency noise is extracted out from $y_t$ using low-frequency cosine transformation.  Projecting $\bar{y}_{+1:m}$ on the space spanned by the first $q$ frequencies is the key to obtaining the conditional distribution of 
$\bar{y}_{+1:m}$. In order to expand the class of processes for which this method can be used while keeping track of parameter uncertainty, \citet{mw16} employed a Bayesian approach. In addition, the resulting PIs are further enhanced to attain the frequentist coverage using the least favorable distribution. This requires advanced algorithmic search for quantiles of non-standard distributions, which is its main drawback in terms of implementation. On the other hand, their supporting online materials provide some pre-computed inputs which make the computation faster.
\begin{enumerate}[(i)]
\item For $q$ small, compute the cosine transformations $ \tran{\bm{x}}=(x_1,\ldots,x_q)$  of series $y_t$.
\item Approximate the covariance matrix of $(\bar{y}_{+1:m},\tran{\bm{x}})$.
\item Solve the minimization problem $(14)$ in \citep[][page 1721]{mw16} to get robust quantiles having uniform coverage.
\item The PIs are given by $[L,U]=\bar{y}+[Q_q^{\textrm{robust}}(\alpha/2),Q_q^{\textrm{robust}}(1-\alpha/2)]$.\\
\end{enumerate}
\textit{Bootstrap PIs for ARX, ETS and NAR}:
\begin{enumerate}[(i)]
\item Adjust $y_t$ for weekly periodicity  using, e.g., seasonal and trend decomposition method proposed by \citet{cca90}. 
\item Perform automatic model selection based on AIC and fit the respective model to adjusted $y_t$. For ARX and NAR, we also use aggregated weather data defined as $\bar{w}_t=\sum_{k=1}^{73}w_{k,t}$, $\bar{\tau}_t=\sum_{l=1}^{78}\tau_{l,t}$ and the weekend-dummy variables as exogenous covariates (see the supplementary Appendix C for details).
\item Simulate $b=1,\ldots,B$ future paths $\hat{y}^b_{n,t}$ of length $m$ from the estimated model.
\item Obtain respective quantiles from set of averages $\bar{\hat{y}}^b_{+,1:m}$,$b=1,\ldots,B$.
\end{enumerate}

\subsection{Data description and goal} We forecast $\bar{y}_{+1:m}=1/m\sum_{t=1}^m y_{n+t}$, i.e. the average of $m$ future  hourly day-ahead spot electricity prices for Germany and Austria - the largest market in the European Power Exchange (EPEX SPOT). One of the reasons why we decided to forecast future averages was that the Bayes approach of \citet{mw16} is designed specifically for the means. Since all other methods are flexible, we used the means as a common basis for the comparison.

The prices arise from day-ahead hourly auctions where traders trade for specific hours of the next day. With the market operating 24 hours a day, we have $11 640$ observations between 01/01/2013 00:00:00 UTC\footnote{Coordinated Universal Time.} and 04/30/2014 23:00:00 UTC. We split the data into a training period spanning from 01/01/2013 00:00:00 UTC to 12/31/2013 23:00:00 UTC and an evaluation period spanning from 01/01/2014 00:00:00 UTC to 04/30/2014 23:00:00 UTC (see Figure \ref{fig:prices}A). The forecasting horizon is $m=1,\ldots,17$ weeks ($168,\ldots,2856$ hours). 

\begin{figure}[!htb]
\includegraphics[scale=0.55, trim={0 8cm  0 8cm },clip]{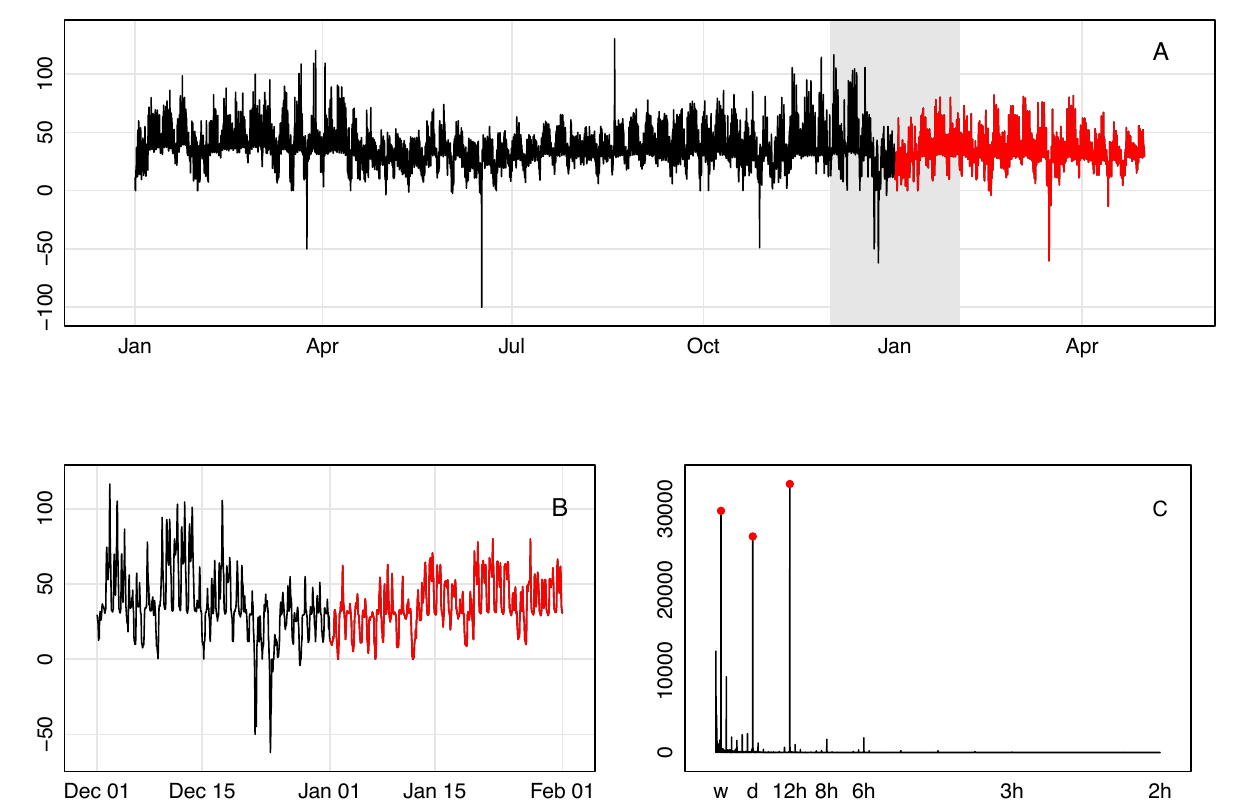}
\centering
\caption{Electricity spot prices, A) Full sample, B) Drop in price level, C) Periodogram with peaks at periods 1 week, 1 day and 12 hours. 
}\label{fig:prices}
\end{figure}

Inspection of the periodogram for the prices in Figure \ref{fig:prices}C reveals peaks at periods  1 week, 1 day and $1/2$ day. The mixed  seasonality is difficult to model by SARIMA or ETS models which are suitable for monthly and quarterly data or by dummy variables. Instead, we use sums of sinusoids with seasonal Fourier frequencies at $\omega_k=2\pi k/168$, $k=1,2,\ldots,\frac{168}{2}$ corresponding to periods 1 week, $1/2$ week, $\ldots$ , 2 hours \citep[see][]{bmrt07,wm08,cf05}. The coefficients of linear combination $\bm \beta_k^{(s)}, \bm \beta_k^{(c)}$ can be estimated by least squares. In addition, we use 2 dummy variables as indicators for all weekends.

As mentioned in Section \ref{sec:intro}, the  local weather variables are also used as covariates. The weather conditions implicitly capture seasonal patterns longer than a week, which is very important when forecasting long horizons. Local weather is represented by 151 hourly wind speed, and temperature series is observed throughout 5 years (2009-2013), i.e., including the training period but not the evaluation period (see above). In order to approximate some missing in-sample data and unobserved values for the evaluation period, we take hourly-specific-averages\footnote{See alternative approximation of future values by bootstrap \citep[in][]{hf10}} of each weather series over these 5 years. In total, we have 168 trigonometric covariates, 151 weather covariates and 2 dummies which gives a full set of 321 covariates.

\subsection{Discussion on the performance of different methods} 
Before we compare the ADJ with the other competitors, we will address a few issues that are usual with analysis of any real-life datasets. In Figure \ref{fig:prices}B, we see a drop in the price level during December 2013. The forecasts based on the whole training period would therefore suffer from bias. By contrast, using only the post-break December data would mean a loss of potentially valuable information. An optimal trade-off in such situations can be achieved by down-weighing older observations \citep[see][]{ppp13}, also called exponentially weighted regression \citep{t10}. In order to achieve better forecasting performance, we use the exponentially weighted regression with standardized exponential weights  $v_{n-t+1}= \delta^{t-1}((1-\delta))/(1-\delta^t )$, $t=1,\ldots,n$ and with $\delta=0.8$. This applies to ADJ and NAR methods. The ETS and ARX models provide exponential down-weighing implicitly, but with optimally selected weights. \citet{mw16} showed that the RBS is robust to structural changes. We would also like to see if there are actual benefits from using disaggregated weather data instead of weather data aggregated across the weather stations. Therefore, we compute the ADJ PIs using  no regressors as in Figure \ref{fig:PI1}A, using  only deterministic regressors as in Figure \ref{fig:PI1}B, using deterministic regressors and aggregated weather variables defined as $\bar{w}_t=\sum_{k=1}^{73}w_{k,t}$, $\bar{\tau}_t=\sum_{l=1}^{78}\tau_{l,t}$ as in Figure \ref{fig:PI1}C and finally, using all 321 covariates as in Figure \ref{fig:PI1}D. As we can see, there is only very little difference between the first three plots, which means that using  only deterministic regressors with or without the aggregated weather data does not prevent the bias at the end of the evaluation period. On the other hand, if we use the disentangled local weather data, significant improvement is achieved. 

\begin{figure}[!htb]
\includegraphics[scale=0.45, trim={0 6cm  0 6cm },clip]{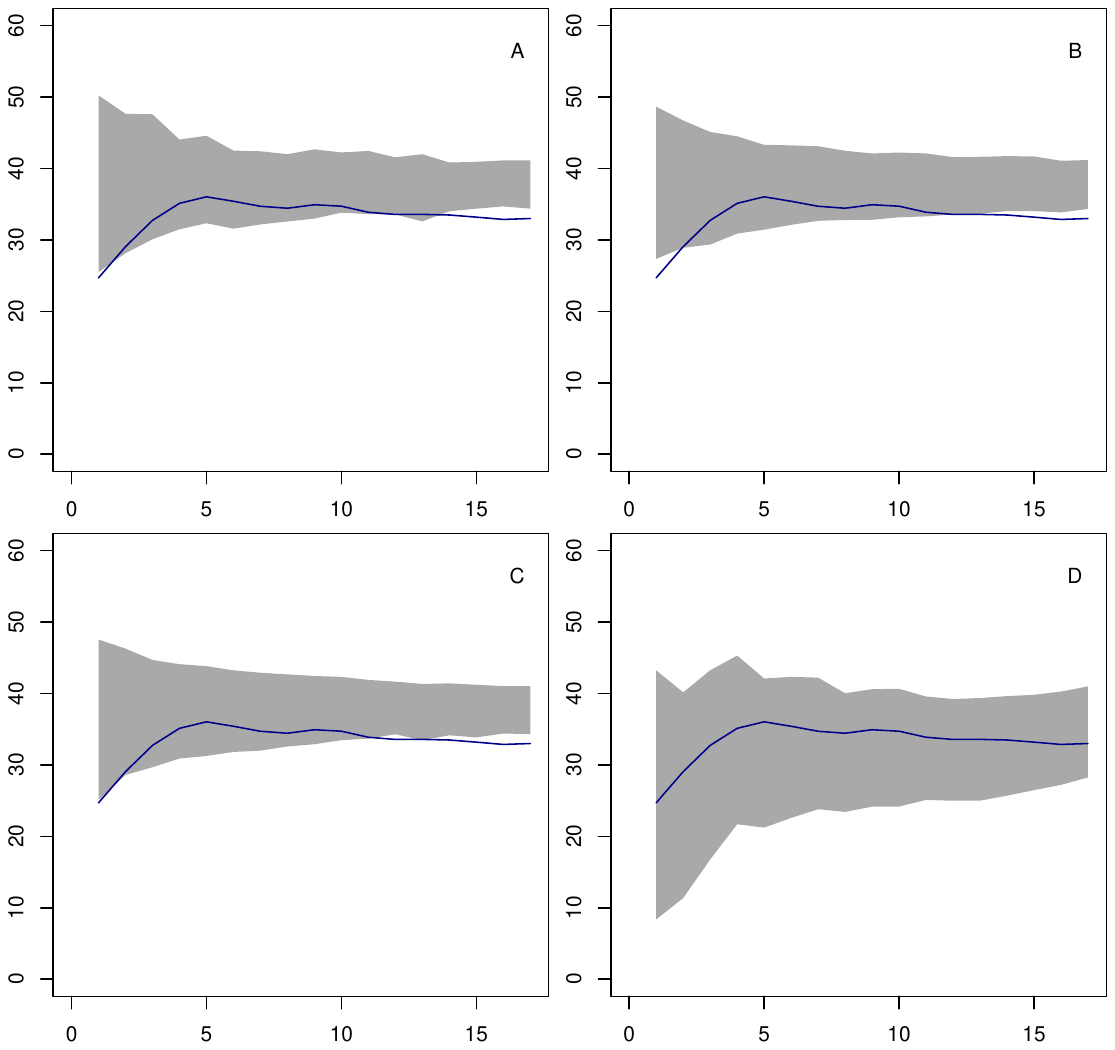}
\centering
\caption{A) ADJ (without covariates), B) ADJ with 170 deterministic covariates, C) ADJ with 170 deterministic covariates \& 2 aggregated (across stations) weather series, D) ADJ with 170 deterministic covariates\& 151 disaggregated weather series.}\label{fig:PI1}
\end{figure}

\begin{figure}[!htb]
\includegraphics[scale=0.45, trim={0 6cm  0 6cm },clip]{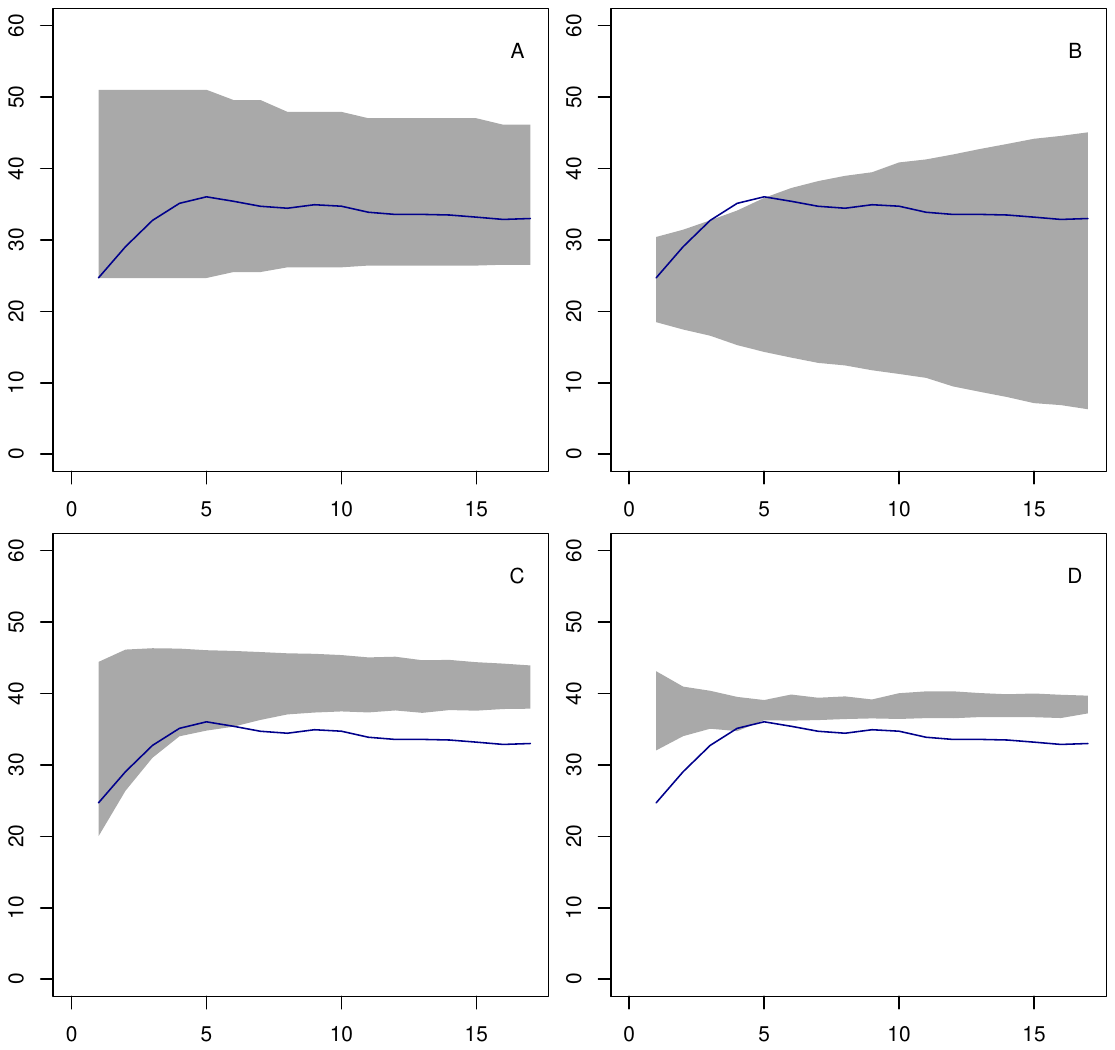}
\centering
\caption{ A) robust Bayes method, B) exponential smoothing state space model (A,N,N) with tuning parameter $0.0446$, C) neural network autoregression (38,22) with one hidden layer D) ARMAX. For C) and D) weekend dummies \& aggregated weather are used as exogenous covariates.}\label{fig:PI2}
\end{figure}

Finally, we get to the comparison with the alternative PIs denoted as RBS, ETS, NAR, and ARX. All these PIs are given in Figure \ref{fig:PI2}. Of the four methods, only RBS gives sensible PIs. RBS works consistently well over the whole $17$-weeks-long evaluation period (Figure \ref{fig:PI2}A). However, when compared to the ADJ, the p.i's seem too conservative. Hence the ADJ provides more precision on top of decent coverage. Prediction intervals by ETS get too conservative as the horizon grows and do not provide a valid alternative to ADJ. The NAR is even more biased than the ADJ without covariates, especially for large $m$. Not so surprisingly, the ARX perform worst of all methods, presumably because the exponential down-weighing implied by the simple autoregression is too mild. Besides, the narrow PIs are the result of ignoring the parameter (among other types of) uncertainty. 


\section{Conclusion}\label{sec:discussion}
\noindent We constructed quantile-based prediction intervals in a regression framework. From a theoretical perspective, we have extended the results of \citet{zhou10} to high-dimensional set-up and also to the case of the nonlinear error process. We showed the consistency of fitted residuals for the ultra-high dimensional case $\log p=o(n)$. Under some mild conditions on the fixed or stochastic covariate process, we were able to establish quantile consistency for the normalized average of fitted residuals and thus provide a significant extension to theoretical validity of the non-parametric and simple quantile-based prediction intervals. The quantile method has been additionally adjusted for short sample and long horizon and was successfully applied to predict spot electricity prices for Germany and Austria using a large set of local weather time series. The results have shown the superiority of the adjusted method over selected conventional methods and approaches such as exponential smoothing, neural networks as well as the recently proposed low-frequency approach of \citet{mw16}.

Regarding future work, some interesting extensions can include multivariate target series and subsequent construction of simultaneous prediction intervals. Applications of such simultaneous intervals could include the prediction of spot electricity prices for each hour simultaneously in the spirit of \citet{rbd15}.



\section{Data Availability Statement}
We are thankful to Stefan Feuerriegel for providing data from their paper \cite{lfn15} for the analysis in a direct communication. 

\section{Acknowledgement}
We are thankful to two anonymous referees and the editor for their helpful comments and corrections that helped improve the presentation of this paper significantly. The first and third authors are partially funded by NSF-DMS 2124222
NSF-DMS 1405410 respectively. 

\ignore{
\renewcommand{\thesubfigure}{\Roman{subfigure}}
\begin{sidewaysfigure}
\vspace{5cm}
\hspace{-2cm} 
\captionsetup[subfigure]{oneside,margin={1cm,-1cm}}
\begin{subfigure}{0.49\hsize}
    \includegraphics[scale=0.5, trim={0 5cm  0 5cm },clip]{Figures/PI_1.pdf}
\caption{A) ADJ (without covariates), B) ADJ with 170 deterministic covariates, C) ADJ with 170 deterministic covariates \& 2 aggregated (across stations) weather time series, D) ADJ with 170 deterministic covariates \& 151 disaggregated weather series.}\label{fig:PI1}
\end{subfigure}%
\hfill\hspace*{0.5cm}  
\begin{subfigure}{0.49\hsize}
    \includegraphics[scale=0.5, trim={0 5cm  0 5cm },clip]{Figures/PI_2.pdf}
\caption{ A) robust Bayes method, B) exponential smoothing state space model (A,N,N) with tuning parameter $0.0446$, C) neural network autoregression (38,22) with one hidden layer D) ARMAX. For C) and D) weekend dummies \& aggregated weather are used as exogenous covariates.}\label{fig:PI2}
\end{subfigure}
\caption{Prediction intervals (gray) for average spot electricity prices (blue) over forecasting horizon $m=1,\ldots,17$ weeks.}\label{fig:PI}
\end{sidewaysfigure}

\clearpage
}

\bibliography{lphdregold}
\newpage
\section*{Appendix A:-Some useful concentration inequalities}\label{sc: linear process results collection}

Throughout, we use properties of the following projection operator  $\mathcal{P_i}$ defined for a sequence of random variable $Y_i$ adapted to a $\sigma$-field filtration $\mathcal{F}_i$ 
\begin{eqnarray}\label{eq:proj}
\mathcal{P}_i Y = \IE(Y |\mathcal{F}_i ) - \IE(Y |\mathcal{F}_{ i-1} ), \quad Y \in \mathcal{L}_1 .
\end{eqnarray}
In particular, one particular inequality that allows an easy transfer from a linear case to a corresponding non-linear case is $$ \|\mathcal{P}_i Y_j\|_p \leq \delta_{i,p}^Y $$
with $\delta_{i,p}^Y$ standing for the dependence measure as defined in (\ref{eq:fdm}) for the mean-zero stationary process $Y_j$. Similar inequality holds for other functional dependence measures we defined in the paper. For a proof, the interested readers can look at the seminal text by \cite{wu05} where functional dependence measure was formally introduced. Apart from this, we have also repeatedly used the orthogonality of the projection operator $\mathcal{P}_i$ defined in (\ref{eq:proj}) i.e. $\IE(\mathcal{P}_i(X_k)\mathcal{P}_j(X_k))=0$ for $ i \neq j$.  With these tools in hand, we proceed to describe the proofs in this paper.

Let $S_{n,b}=\sum_{i=1}^n b_i e_i$.  Only the light-tailed versions of the Results \ref{th:nagaev linear light} and \ref{th:nagaev nonlinear light} were proposed and proved in \citet{wuwu16}. In this exposition, we clearly point out the corresponding results for the long-range and short-range dependence along with the new addition of what happens for light-tailed or heavy-tailed scenarios. These results can be of independent interest since this is first such in the literature that takes care of heaviness of tails under dependence.

\begin{result}\label{th:nagaev linear light}
(Nagaev inequality for linear processes)
Assume that the error process $e_i$ admits the representation (\ref{eq:ei linear}). Then we have the following concentration results for $S_{n,b}=\sum_{i=1}^n b_ie_i$,
\begin{compactitem}[-]
\item  Light-tailed SRD: If $\sum_{j}|a_j|<\infty$ and $\epsilon_j \in \mathcal{L}^q$ for some $q>2$, then, for some constant $c_q$,
\begin{eqnarray*}\label{eq:linear light}
\IP(|S_{n,b}| \geq x) &\leq& (1+2/q)^q \frac{\|b\|_q^q (\sum_{j}|a_j|)^q\|\epsilon_0\|_q^q}{x^q}\nonumber+ 2 \exp \left(- \frac{c_qx^2}{n(\sum_{j}|a_j|)^2 \|\epsilon_0 \|_2^2} \right),
\end{eqnarray*}
\item Light-tailed LRD: If $K=\sum_{j}|a_j|(1+j)^{\bm \beta}<\infty$ for $0<\bm \beta<1$ and $\epsilon_j \in \mathcal{L}^q$ for some $q>2$, then, for some constant $C_1,C_2$ depending on only $q$ and $\bm \beta$,
\begin{eqnarray}\label{eq:linear light long}
\IP(|S_{n,b}| \geq x) &\leq& C_1 \frac{K^q \|b\|_q^q\|n^{q(1-\bm \beta)}\epsilon_0\|_q^q}{x^q} 
+ 2 \exp \left(- \frac{C_2x^2}{n^{3-2\bm \beta}\|\epsilon_0 \|_2^2K^2} \right),
\end{eqnarray}
\item Heavy-tailed SRD: If $\sum_{j}|a_j|<\infty$ and $\epsilon_j \in \mathcal{L}^q$ for some $1<q \leq2$, then, for some constant $c_q$
\begin{equation}\label{eq:linear heavy}
\IP(|S_{n,b}| \geq x) \leq c_q \frac{\|b\|_q^q (\sum_{j}|a_j|)^q\|\epsilon_0\|_q^q}{x^q}, 
\end{equation}
\item Heavy-tailed LRD: If $K=\sum_{j}|a_j|(1+j)^{\bm \beta}<\infty$ for $0<\bm \beta<1$ and $\epsilon_j \in \mathcal{L}^q$ for some $q>2$, then, for some constants $C_1,C_2$ depending only on $q$ and $\bm \beta$,
\begin{equation}\label{eq:linear heavy long}
\IP(|S_{n,b}| \geq x) \leq C_1 \frac{K^q \|b\|_q^q\|n^{q(1-\bm \beta)}\epsilon_0\|_q^q}{x^q}.
\end{equation}
\end{compactitem}
\end{result}

\noindent The proofs of the Theorem \ref{th:lighttailed} and Theorem \ref{th:light tailed linear} heavily rely on the representation in (\ref{eq:ei linear}). Therefore it is important to explore a new approach to prove the analogous result for nonlinear processes.

\ignore{We use the following result for the case of finitely many covariates
\begin{theorem}[Residual consistency for regression]
$$\sum_i |\hat{e}_i-e_i|= o_{P}(\Pi(n))$$
The error bound $\Pi(n)$ will be different for the different behavior of the error process $e_i$.
\end{theorem}
}

\ignore{Before we present quantile consistency results for the case where the number of predictors grows much faster than the sample size, we  state key optimal tail-probability inequalities needed for the two cases, linear and nonlinear error process. }

\begin{result} \label{th:nagaev nonlinear light}
(Nagaev inequality for nonlinear processes)
Assume $\|e. \|_{q, \alpha}<\infty$ for some $\alpha>0$,
\begin{compactitem}[-]

\item Light-tailed SRD:- Assume that $\|e.\|_{q,\alpha} <\infty$ where $q>2$ and $\alpha>0$ and $\sum_{i=1}^n b_i^2=n.$ Let $r_n=1$(resp. $(\log n)^{1+2q}$ or $n^{q/2-1-\alpha q}$ ) if $\alpha>1/2-1/q$ (resp. $\alpha=0$ or $\alpha< 1/2-1/q$). Then for all $x>0$, for constants $C_1,C_2,C_3$ that depend on only $q$ and $\alpha$,
\begin{eqnarray}\label{eq:nonlinear light}
\IP(|S_{n,b}| \geq x) \leq C_1 \frac{r_n}{(\sum_{j}|b_j|)^q\|e.\|_{q,\alpha}^q}{x^q}+C_2 \exp \left ( -\frac{C_3x^2}{n \|e.\|^2_{2,\alpha}}\right) ,
\end{eqnarray} 

\item Heavy-tailed SRD:- Assume that $\|e.\|_{q,\alpha} <\infty$ where $1<q<2$ and $\alpha>0$ and $\sum_{i=1}^n b_i^2=n.$ Let $r_n=1$(resp. $(\log n)^{1+2q}$ or $n^{q/2-1-\alpha q}$ ) if $\alpha>1/2-1/q$ (resp. $\alpha=0$ or $\alpha< 1/2-1/q$). Then for all $x>0$, for constants $C_1$ that depend on only $q$ and $\alpha$,
\begin{equation}\label{eq:nonlinear heavy}
\IP(|S_{n,b}| \geq x) \leq C_1 \frac{r_n}{(\sum_{j}|b_j|)^q\|e.\|_{q,\alpha}^q}{x^q}.
\end{equation} 
\end{compactitem}
\end{result}

\begin{proof}[Sketch for the proof of Result \ref{th:nagaev linear light} and \ref{th:nagaev nonlinear light}]
The details are omitted for these proofs since the steps are similar to that of Theorem 1 and Theorem 2 in \citet{wuwu16}. Very briefly, the basic strategy lies in writing $S_{n,b}$ with the following decomposition

\begin{eqnarray}
S_{n,b}=S_{n,b,0}+(S_{n,b}-S_{n,b,n})+\sum_{l=1}^{\lceil \log n/ \log 2)\rceil} \sum_{k=1}^{n}b_k(e_{k,\min (n,2^l)}-e_{k,2^{l-1}})\nonumber,
\end{eqnarray}
where $e_{i,\tau}= \IE(e_i|\epsilon_{i-\tau},\ldots, \epsilon_i)$ and $S_{n,b,m}=\sum_{k=1}^n b_k e_{k,m}.$ 

With this decomposition, the summands in the first term are independent and hence Nagaev inequality in Corollary 1.7 from \citet{MR542129} can be applied. For the second term Burkholder and H\"older inequalities are used whereas for the third term due to $2^l$-dependence we divide the summands in independent odd and even parts and proceed with Corollary 1.7 from \citet{MR542129} again. For the heavy-tailed scenarios, note that the exponential terms from the light-tailed are not presented separately for the interest of space since one just needs to use corollary 1.6 instead of corollary 1.7 from \cite{MR542129}.  

\end{proof}

\noindent The proof for the quantile consistency results for the lasso-fitted residuals in Section \ref{sec:high dimension} requires the following important result.

\begin{lemma}\label{eq:bickel lemma}
Assume the the restricted eigenvalue condition (\ref{eq:kappa}) with $\kappa$ as presented in Theorem \ref{th:covariate2}. Then on the event $\mathcal{A}= \bigcap_{j=1}^p \{2|V_j| \le r \} , \text{ where } V_j =\frac{1}{n}\sum_{i=1}^n e_ix_{ij},$ we have,

\begin{eqnarray}\label{eq:toshow}
\|X(\hat{\bm \beta}-\bm \beta)\|_2^2/n \leq 16 sr^2/\kappa^2,\quad \|\hat{\bm \beta}-\bm \beta\|_1\leq 4 sr/\kappa^2.
\end{eqnarray}

\end{lemma}

\begin{proof}[Proof of Lemma \ref{eq:bickel lemma}]
Since $\hat{\bm \beta}$ minimizes the LASSO-objective in (\ref{eq:lasso objective}), plugging in $y_i=\tran{\bm{x}}_i\bm \beta+e_i$ and $\lambda=2r$, we have

\begin{eqnarray*}
2r\|\hat{\bm \beta}\|_1 +\frac{1}{n}\sum_{i=1}^{n}(\tran{\bm{x}}_i\bm \beta+e_i-\tran{x}_i \hat{\bm \beta})^2
&\leq& 2r\|\bm \beta\|_1+\frac{1}{n}\sum_{i=1}^{n}e_i^2 \\
2r\|\hat{\bm \beta}\|_1+\|X(\hat{\bm \beta}-\bm \beta)\|_2^2/n 
&\leq& 2r\|\bm \beta\|_1+2\tran{e}X(\hat{\bm \beta}-\bm \beta)/n \\
r\|\hat{\bm \beta}-\bm \beta\|_1+\|X(\hat{\bm \beta}-\bm \beta)\|_2^2/n 
&\leq& 2r(\|\bm \beta\|_1-|\hat{\bm \beta}\|_1)+r\|\hat{\bm \beta}-\bm \beta\|_1 +2\tran{e}X(\hat{\bm \beta}-\bm \beta)/n.
\end{eqnarray*}

\noindent These entail, on the event $\{|V_j|= |(\tran{e}X)_j|/n \leq r/2\text{ for all }j \}$ 
\begin{eqnarray}\label{eq:final1}
&&r\|\hat{\bm \beta}-\bm \beta\|_1+\|X(\hat{\bm \beta}-\bm \beta)\|_2^2/n \leq 4r \|\hat{\bm \beta}_J- \bm \beta_J\|_1\leq 4r\sqrt{s} \|\hat{\bm \beta}_J- \bm \beta_J\|_2,
\end{eqnarray}
where $J=\{j:\beta_j\neq 0\}$ as for $j \in J^{c}$, $|\hat{\beta}_j-\beta_j|+| \beta|_j-|\hat{\beta}|_j =0$.
However, also note that from the first inequality in (\ref{eq:final1}), we have $\|\hat{\bm \beta}_{J^c}- \bm \beta_{J^c}\|_1 \leq 3\|\hat{\bm \beta}_J- \bm \beta_J\|_1.$ Thus by (\ref{eq:kappa}), $ \kappa\|\hat{\bm \beta}_J- \bm \beta_J\|_2 \leq \|X(\hat{\bm \beta}-\bm \beta)\|_2/\sqrt{n}$. Plugging this in (\ref{eq:final1}) immediately yields the first inequality of (\ref{eq:toshow}) whereas the second is derived by square completion.

\end{proof}

\begin{lemma}\label{lem:x'x}
Recall the dependence adjusted and dimension incorporated functional dependence measure as defined in the context of Theorem \ref{th:lasso stochastic}. Assume $\Phi_{q,\alpha} <\infty$ for some $q \geq 1$ and $\alpha>0$. For a fixed $\delta>0$, assume $p,n \to \infty$ with $2C_{\alpha}\Phi_{4,\alpha}^4s^2 \log p < n\delta^2$. Then,

$$\IP(|\tran{X}X/n-\IE(\bm x_t \tran{\bm x_t})|_{\infty} >\delta/s) \to 0.$$
\end{lemma}

\begin{proof}[Proof of Lemma \ref{lem:x'x}]
We prove this by applying Hansen-Wright inequality as presented in \cite{zhangwu} on quadratic forms. Although they proved it for a more general locally stationary process and defined their functional dependence measure accordingly, it also holds for the simpler (not accounting for local stationarity) functional dependence measure defined in the context of Theorem \ref{th:lasso stochastic}. Apply Theorem 6.4 in \cite{zhangwu} for $B=0$, $x=n\delta/s$, $a_{t}=1$ to obtain

\begin{eqnarray*}
\IP(|\tran{X}X- n\IE(\bm x_t \tran{\bm x_t})|_{\infty}>n\delta/s )&\leq& Cp^2 \exp \left(- \frac{n \delta^2}{s^2 C_{\alpha}\Phi_{4,\alpha}^4} \right)
\end{eqnarray*}
which goes to $0$ due to our assumption on how $p,s$ and $n$ scale with each other.

\end{proof}


\section{Appendix B:- Proofs of Theorems}\label{sec:proofs}
 
\begin{proof}[Proof of Theorem \ref{th:quenched clt}]

The $m$-dependence approximation is the key idea for the proof for the nonlinear case, 
\begin{eqnarray*}
\|\IE(\tilde{S}_m|\mathcal{F}_0)-\IE (S_m|\mathcal{F}_0) \| \leq \|S_m-\tilde{S}_m\| \leq m^{1/2}\Theta_{m,p} \ll m^{1/2}/(\log m)^{5/2},
\end{eqnarray*}
where $\tilde{S}_m=\sum_{i=1}^m \tilde{e}_i=\sum_{i=1}^m \IE(e_i|\epsilon_{i}, \ldots \epsilon_{i-m}).$ Note the following decomposition 
\begin{eqnarray*}
\IE(\tilde{S}_m|\mathcal{F}_0)&=& \sum_{j=-m}^0 \mathcal{P}_j \tilde{S}_m
=\sum_{j=-\infty}^0 (\IE(\tilde{S}_m|\mathcal{F}_j)- \IE(\tilde{S}_m|\mathcal{F}_{j-1})).
\end{eqnarray*}
The proof of the sufficient condition for convergence (\ref{eq:sm|f0 cond}) then follows from $\|\mathcal{P}_j \tilde{e}_i\|_2 \leq \delta_{i-j,2}$ and the condition on $\Theta_{m,p}$.

\end{proof}

\noindent For the rest of the proofs we introduce some notations. Define 
\begin{eqnarray}\label{eq:tilde Y}
\tilde{Y}_i= H_m^{-1}\sum_{j=i-m+1}^i e_j \text{ for }i=m,m+1,\ldots,
\end{eqnarray}
and let $\tilde{Z}_i= \tilde{Y}_i- \IE(\tilde{Y}_i|\mathcal{F}_{i-1})$.
\ignore{
Note that one can write $\tilde{Z}_i$ as follows

\begin{equation}\label{eq:zi}
\tilde{Z}_{i-1}= \frac{\sum_{j=1}^{\infty} \tilde{b}_j \epsilon_{i-j}}{H_m}
\end{equation}

\noindent where $\tilde{b}_j= a_0+ a_1+ \ldots+ a_j$ if $1 \leq j \leq m-1$ and $\tilde{b}_j= a_{j-m+1}+a_{j-m+2}+ \ldots+ a_j$ if $j \geq m$. 
}
\noindent Define $$ \tilde{F}_n^*(x)= \frac{1}{n-m+1}\sum_{i=m}^n \IP(\tilde Y_i \leq x).$$

\noindent Let $\tilde{F}(x)=\IP(\tilde{Y}_i \leq x).$ Let $\tilde{F}_n(x)$ denote the empirical distribution function of $\tilde{Y}_i, i=m,\ldots n$. We use the following decomposition:

\begin{eqnarray}\label{eq:decomposition2}
\tilde{F}_n(x)-\tilde{F}(x)=M_n(x)+N_n(x)
\end{eqnarray}
\noindent where $M_n(x)=\tilde{F}_n(x)-\tilde{F}_n^*(x)$ and $N_n(x)=\tilde{F}_n^*(x)-\tilde{F}(x)$. In the next series of lemma, we prove some consistency results involving $M_n(\cdot)$ and $N_n(\cdot)$. These are fundamental in the eventual proof of quantile consistency results from Section \ref{sec:nonlinear}. The proof techniques are similar to that used in \citet{wuzhou09}, \cite{zhou10} adapted to the nonlinear setting and specifically exploiting the predictive density dependence defined in (\ref{eq:pred density}). Using the projection operator defined in (\ref{eq:proj}), one can write $M_n(x)$ as follows
\begin{equation}\label{eq:define Mn}
M_n(x)=\frac{1}{n-m+1} \sum_{i=m}^n \mathcal{P}_i I (\tilde Y_i \leq x).
\end{equation}
Next we present two important lemmas concerning local equicontinuity of the two terms $M_n(\cdot)$ and $N_n(\cdot)$. Let $f_{\epsilon}$ is the density of the conditional distribution of $\tilde{Y}_i$ given $\mathcal{F}_{i-1}$.

\begin{lemma}\label{lem:Mn} Under conditions of Theorem \ref{th:nagaev linear light} and Theorem \ref{th:nagaev nonlinear light},
\begin{eqnarray}\label{eq:Mn bound}
\sup_{|u|\leq b_n} &&|M_n(x+u)-M_n(x)|
=O_{\IP} \left( \sqrt{\frac{H_m b_n}{n}} \log^{1/2}n +n^{-3} \right)\nonumber,
\end{eqnarray}
\noindent where $b_n$ is a positive bounded sequence with $\log n=o(H_mnb_n).$ 
\end{lemma}

\begin{proof}
Note that, $\IP(x \le \tilde{Y}_i \le x+u|\mathcal{F}_{i-1})\le H_mc_0u$ for all $u>0$ where $c_0=\sup_x|f_{\epsilon}(x)|<\infty$. Therefore, the result in (\ref{eq:Mn bound}) follows by Freedman's martingale inequality and a chaining argument as for any $u \in [-b_n,b_n]$, we have 

\begin{equation}\label{eq:}
\sum_{i=m}^{n}[\IE(V_i)-\IE(V_i)^2] \le c_0(n-m+1)H_mb_n 
\end{equation}

\noindent where $ V_i=I(x\leq \tilde{Y}_i \leq x+u|\mathcal{F}_{i-1})$. We skip the details and direct interested readers to Lemma 5 in \cite{wubahadur05}, Lemma 4 in \cite{wumest07} and Lemma 6 in \cite{wuzhou09}. 

\end{proof}

\begin{lemma}\label{lem:Nn th1}
Under conditions of SRD, DEN and light-tailed
\begin{equation}\label{eq:lemma Nn Th1 equation}
\| \sup_{|u|\leq b_n} |N_n(x+u)-N_n(x) |\| =O \left(  \frac{b_n m^{3/2}}{\sqrt{n}} \right).
\end{equation}

\end{lemma}

\begin{proof}
\noindent From the definition of $N_n(x)$ in (\ref{eq:decomposition2}), we have 

\begin{eqnarray}\label{eq:Nn}
N_n(x+u)-N_n(x)= \sqrt{m} \frac{\int_0^u R_n(x+t)dt}{n-m+1},
\end{eqnarray}   where, for $x \in \mathbb{R}$, 
\begin{eqnarray}\label{eq:Nn2}
R_n(x)=\sum_{i=m}^n [f_{\epsilon}(H_m(x-\tilde{Z}_{i-1}))-\IE(f_{\epsilon}(H_m(x-\tilde{Z}_{i-1})))].
\end{eqnarray}

\noindent Let $(\epsilon_i')_{-\infty}^{\infty}$ be an i.i.d. copy of $(\epsilon_i)_{-\infty}^{\infty}$. Let  $Z^{*}_{i-1,k}=H(\epsilon_i,\epsilon_{i-1},\ldots)$. Denote  $\tilde{Z}^{*}_{i-1,k}=H(\epsilon_i,\epsilon_{i-1},\ldots,\epsilon_{i-k}',\ldots)$. Also, introduce the coefficients $\tilde{b}_{j,q}$ as follows

\begin{equation}\label{eq:bj}
\tilde{b}_{j,q}=  \begin{cases}
            & \psi_{0,q}+\psi_{1,q}+\ldots+\psi_{j,q} \quad \text{if } 1\le j \le m-1 \\
            & \psi_{j-m+1,q}+\psi_{j-m+2,q}+\ldots+\psi_{j,q} \quad \text{if }j \ge m.
            \end{cases}
\end{equation}
Write $\tilde{b}_j$ for $\tilde b_{j,2}$. Then

\begin{eqnarray}\label{eq:36}
\|\mathcal{P}_{i-k}f_{\epsilon}(\sqrt{m}(x+u-\tilde{Z}_{i-1})) \| &\le& \|f_{\epsilon}(\sqrt{m}(x+u-\tilde{Z}_{i-1}))-f_{\epsilon}(\sqrt{m}(x+u-\tilde{Z}^*_{i-1,k})) \| \nonumber  \\
& \le & \sup_{v \in \mathbb{R}} |f_{\epsilon}'(v)| \| \sqrt{m}(\tilde{Z}_{i-1}-\tilde{Z}^*_{i-1.k}) \| \le c_1 \tilde{b}_k,
\end{eqnarray}

\noindent for some $c_1<\infty.$ Further note that

$$R_n(x+u)=\sum_{k=1}^{\infty}\sum_{i=m}^n \mathcal{P}_{i-k} f_{\epsilon}(\sqrt{m}(x+u-\tilde{Z}_{i-1}))$$

\noindent and by the orthogonality of $\mathcal{P}_{i-k}, i=m,\ldots, n$

\begin{eqnarray*}\label{eq:projection}
\vspace{-0.2 in}\|\sum_{i=m}^n \mathcal{P}_{i-k} f_{\epsilon}(\sqrt{m}(x+u-\tilde{Z}_{i-1})) \|^2 &=& \sum_{i=m}^n \|\mathcal{P}_{i-k} f_{\epsilon}(\sqrt{m}(x+u-\tilde{Z}_{i-1})) \|^2 \nonumber \\
&\leq& c_1^2 (n-m+1)\tilde{b}_k^2. 
\end{eqnarray*}

\noindent Therefore, for all $u \in [-b_n,b_n]$,  
\begin{eqnarray*}
\|R_n(x+u)\| &\le& \sum_{k=1}^{\infty} \| \sum_{i=m}^n \mathcal{P}_{i-k}f_{\epsilon}(\sqrt{m}(x+u-\tilde{Z}_{i-1})) \| \\
&\le& c_1 \sqrt{n} \sum_{k=1}^{\infty} |\tilde{b}_k| \le c_1 m\sqrt{n} \sum_{j=0}^{\infty}|\psi_{j,2}|=O(m\sqrt{n})
\end{eqnarray*}
in the view of the SRD condition in (\ref{eq:conditions2}). This concludes the proof of (\ref{eq:lemma Nn Th1 equation}) in the light of (\ref{eq:Nn}) and (\ref{eq:Nn2}).

\end{proof}

\begin{lemma}\label{lem:Nn th4}
\noindent Under conditions of LRD, DEN and heavy-tailed, we have for any $\rho \in (1/\gamma, \alpha)$
\begin{equation}\label{eq:lemma Mn Th1 equation}
\| \sup_{|u|\leq b_n} |N_n(x+u)-N_n(x) |\|_{\rho} =O \left(  H_mb_nmn^{1/\rho-\gamma}|l(n)| \right).
\end{equation}

\end{lemma}

\begin{proof}
\noindent Similar to the proof of Lemma \ref{lem:Nn th1}, it suffices to prove, for some $0<C<\infty$,

\begin{equation}\label{eq:lemma Nn th4 suff}
\|R_n(x+u) \|_{\rho} \le Cmn^{1/\rho+1-\gamma}|l(n)| \text{ for all } u \in [-b_n,1-b_n]
\end{equation}

\noindent Since $1<\rho<2$, by (\citet{rio}) Burkholder type inequality of martingales, we have, with $C_{\rho}=(\rho-1)^{-1}$. 

\begin{eqnarray}\label{eq:Nn lemma 2}
\|R_n(x+u)\|_{\rho}^{\rho} 
&=& \|\sum_{k=-\infty}^{n-1} \mathcal{P}_k \sum_{i=m}^n f_{\epsilon}(H_m(x-\tilde{Z}_{i-1})) \|_{\rho}^{\rho} \nonumber \\ \nonumber
&\le& C_{\rho} \sum_{k=-\infty}^{n-1}\| \mathcal{P}_k \sum_{i=m}^n f_{\epsilon}(H_m(x-\tilde{Z}_{i-1})) \|_{\rho}^{\rho} \\ \nonumber
&\le&C_{\rho} \sum_{k=-\infty}^{n-1} (\sum_{i=m}^n \|\mathcal{P}_k  f_{\epsilon}(H_m(x-\tilde{Z}_{i-1})) \|_{\rho})^{\rho} \\ \nonumber
&\le&C_{\rho}  \sum_{k=-\infty}^{-n}+\sum_{k=-n+1}^{0}+\sum_{k=1}^{n-1}(\sum_{i=m}^n \|\mathcal{P}_k  f_{\epsilon}(H_m(x-\tilde{Z}_{i-1})) \|_{\rho})^{\rho} \\ \nonumber
&\le&C_{\rho}(I+II+III) \nonumber.
\end{eqnarray}

\noindent Since $\IE|\epsilon_i|^{\rho}<\infty$, similarly as (\ref{eq:36}), we have for $k \le i-1$ that 

\begin{equation}\label{eq:Nn lemma 2 second}
\|\mathcal{P}_kf_{\epsilon}(H_m(x-Z_{i-1})) \|_{\rho} \le c_1|\tilde{b}_{i-k}|,
\end{equation} 
 for some $c_1<\infty.$ Thus using Karamata's theorem for the term $I$, we have 

\begin{eqnarray}\label{eq:term 1}
I&\le& c_1^{\rho}\sum_{k=-\infty}^{-n} (\sum_{i=m}^n|\tilde{b}_{i-k}|)^{\rho} \leq c_1^{\rho}\sum_{k=n}^{\infty} (m \sum_{i=1}^n |\psi_{k+i,\rho}|)^{\rho} \\ \nonumber 
&\le& c_1^{\rho}m^{\rho}n^{\rho-1}\sum_{k=n}^{\infty}\sum_{i=1}^n |\psi_{k+i,\rho}|^{\rho} \\ \nonumber 
&=&O[m^{\rho}n^{1+\rho(1-\gamma)}|l(n)|^{\rho}].
\end{eqnarray}

\noindent Since $\rho>1$ and $\rho \gamma>1$, we use H{\"o}lder inequality to manipulate term $III$ as follows:

\begin{eqnarray}\label{eq:term 3}
III &\le& c_1^{\rho} \sum_{k=1}^{n-1} (\sum_{i=\max(m,k+1)}^n |\tilde{b}_{i-k}|)^{\rho} 
\le c_1^{\rho} \sum_{k=1}^{n-1} (m \sum_{i=0}^{n-k}|\psi_{i,\rho}|)^{\rho} \\ \nonumber
&=& m^{\rho}\sum_{k=1}^{n-1} O[(n-k)^{1-\gamma}|l(n-k)|]^{\rho} \\ \nonumber
&=& O[m^{\rho}n^{1+\rho(1-\gamma)}|l(n)|^{\rho}].
\end{eqnarray}

\noindent Similarly for term $II$ we have, $II=O[m^{\rho}n^{1+\rho(1-\gamma)}|l(n)|^{\rho}]$. Combining this with (\ref{eq:term 1}) and (\ref{eq:term 3}), we finish the proof of the lemma.
\end{proof}

\begin{proof}[Proof of Theorem \ref{th:light tailed nonlinear}]
\noindent Recall the definition of $\tilde{Y}_i$ from (\ref{eq:tilde Y}). As $m \to \infty$, $\tilde{Q}(u)$ is well-defined and it converges to $u$th quantile of a $N(0,\sigma^2)$ distribution since the central limit theorem of \cite{ref22zxw} entails $\tilde{Y}_i \stackrel{D}{\rightarrow} N(0,\sigma^2),$ where $\sigma=\|\sum_{i=0}^{\infty}\mathcal{P}_0e_i \|<\infty.$  A standard characteristic function argument yields  

\begin{equation}\label{eq:37}
\sup_x|\sigma f_m(x)- \phi(x/\sigma)| \to 0,
\end{equation}

\noindent where $f_m(\cdot)$ is the density of $\tilde{Y}_i$ and $\phi(x)$ is the density of a standard normal random variable. Let $(c_n)$ be an arbitrary sequence of positive numbers that goes to infinity. Let $\bar{c}_n= \min (c_n, n^{1/4}/m^{3/4}).$ Then $\bar{c}_n \to \infty.$ For $T_n=\bar{c}_nm/\sqrt{n}$, Lemma \ref{lem:Mn} and Lemma \ref{lem:Nn th1} imply that 

\begin{eqnarray}
&&|\tilde{F}_n(\tilde{Q}(u)+T_n)-\tilde{F}(\tilde{Q}(u)+T_n) -[F_n(\tilde{Q}(u))-\tilde{F}(\tilde{Q}(u))]| \nonumber \\ && \quad\quad=O_{\IP}\left(\frac{T_nm^{3/2}}{\sqrt{n}}+m^{1/4}\sqrt{\frac{T_n}{n}}(\log n)^{1/2} \right)=o_{\IP}(T_n).\label{eq:38}
\end{eqnarray}

\noindent Similar arguments as those in Lemma \ref{lem:Mn} and Lemma \ref{lem:Nn th1} imply 

\begin{equation}\label{eq:39}
|\tilde{F}_n(\tilde{Q}(u))-\tilde{F}(\tilde{Q}(u))|=O_{\IP}(\frac{m}{\sqrt{n}})=o_{\IP}(T_n).
\end{equation}

\noindent  Using Taylor's expansion of $\tilde{F}(\cdot)$, we have 

\begin{equation}\label{eq:40}
\tilde{F}(\tilde{Q}(u)+T_n)-\tilde{F}(\tilde{Q}(u))=T_n f_m(\tilde{Q}(u))+O(T_n)^2.
\end{equation}

\noindent  By (\ref{eq:37}), $f_m(\tilde{Q}(u))>0$ for sufficiently large $n$. Plugging in (\ref{eq:39}) and (\ref{eq:40}) into (\ref{eq:38}), we have $\IP(\tilde{F}_n(\tilde{Q}(u)+T_n)>u) \to 1$. Hence $\IP(\hat{Q}_n(u)>\tilde{Q}(u)+T_n) \to 0$ by the monotonicity of $\tilde{F}_n(\cdot).$ Similar arguments yield $\IP(\hat{Q}_n(u)<\tilde{Q}(u)-T_n) \to 0$. Since $f_m(\tilde{Q}(u))>0$ and $T_n \to 0$ arbitrarily slow, Theorem \ref{th:light tailed nonlinear} is proved.

\end{proof}

\noindent Next we use Lemma \ref{eq:bickel lemma} along with Nagaev inequalities from Section \ref{sec:noregress} to prove Theorem \ref{th:covariate2} and Theorem \ref{th:lasso stochastic}.

\begin{proof} [Proof of Theorem \ref{th:covariate2}]
In the view of Lemma \ref{eq:bickel lemma} and appropriate Nagaev inequality from  Result \ref{th:nagaev nonlinear light}, we have, 

\begin{eqnarray}
\vspace{-0.5 in}&&\IP\left(\frac{1}{n}\|X(\hat{\bm \beta}-\bm \beta)\|^2_2 \geq 16sr^2/\kappa^2\right) 
\leq \sum_{j=1}^p \IP(2|V_j|>r)=o(1) \nonumber
\end{eqnarray}

\noindent due to the choice of $r$ in (\ref{eq:r}). Thus, 
\begin{eqnarray}\label{eq:eihat-ei}
\sup_{m \leq i \leq n}|\sum_{k=i-m+1}^{i}(\hat{e}_i-e_i)|\leq m \sup_{1 \leq i \leq n}|\hat{e}_i-e_i| 
\leq m \sqrt{\frac{1}{n}\|X(\hat{\bm \beta}-\bm \beta)\|^2_2} =O_{\IP}(m\sqrt{s}r) \nonumber.
\end{eqnarray}
\noindent The rest of the proof follows from the following observation; for any fixed $0 \leq u \leq 1,$

\begin{equation}\label{eq:newbound}
|\bar{Q}_n(u)-\hat{Q}_n(u)|=O_{\IP} \left( \frac{m\sqrt{s}r}{H_m} \right),
\end{equation}
where $H_m$ is properly chosen for the heavy or light tails as described in (\ref{eq:Hm}). Then the right hand side of (\ref{eq:newbound}) by the choice of $s$ as specified in (\ref{eq:sparsity condition}) are smaller than the upper bounds of the SRD specific cases mentioned in Theorem \ref{th:light tailed nonlinear}.


\end{proof}

\begin{proof} [Proof of Theorem \ref{th:lasso stochastic}]
Consider the events $$B_1=\{n^{-1}|\tran{e}X|_{\infty}<r/2\}\text{ and }B_2=\{|\tran{X}X/n-\IE(\bm x_t \tran{\bm x_t})|_{\infty} <\kappa_{stoch}^2/(32s)\}.$$
Under $B_2$, the following holds for a vector $v$ with   $|v_{J^c}|_1 \leq 3 |v_{J}|_1$,
\begin{eqnarray*}
\frac{\tran{v} \tran{X}X v}{ n\|v_J\|_2^2}&=& \frac{\tran{v} (\tran{X}X/n - \IE(\bm x_t \tran{\bm x_t})) v}{ \|v_J\|_2^2}+ \frac{\tran{v} \IE(\bm x_t \tran{\bm x_t}) v}{ \|v_J\|_2^2} \\
&\geq&  -\frac{|v|_1|(\tran{X}X/n - \IE(\bm x_t \tran{\bm x_t})) v|_{\infty}}{ \|v_J\|_2^2}++ \frac{\tran{v} \IE(\bm x_t \tran{\bm x_t}) v}{ \|v_J\|_2^2} \\
&\geq& -\frac{|v|_1^2|\tran{X}X/n - \IE(\bm x_t \tran{\bm x_t})|_{\infty}}{ \|v_J\|_2^2}++ \frac{\tran{v} \IE(\bm x_t \tran{\bm x_t}) v}{ \|v_J\|_2^2} \\
&\geq& -\frac{16 |v_{J}|_1^2|\tran{X}X/n - \IE(\bm x_t \tran{\bm x_t})|_{\infty}}{ \|v_J\|_2^2}++ \frac{\tran{v} \IE(\bm x_t \tran{\bm x_t}) v}{ \|v_J\|_2^2} \\
&\geq& -\frac{16 \left(\frac{\kappa_{stoch}^2}{32s} \right) |v_{J}|_1^2 }{ \|v_J\|_2^2}+ \frac{\tran{v} \IE(\bm x_t \tran{\bm x_t}) v}{ \|v_J\|_2^2} \\
&\geq&- \frac{\kappa_{stoch}^2}{2}+\frac{\tran{v} \IE(\bm x_t \tran{\bm x_t}) v}{ \|v_J\|_2^2}.
\end{eqnarray*}
In the light of  (\ref{eq:kappastoch}), this implies $$ \min_{J \subset \{1,\cdots,p\},|J|\leq s,} \min_{|u_{J^c}|_1 \leq 3 |u_{J}|_1}\frac{\tran{u} \tran{X}X u}{n \|u_J\|_2^2}  > \frac{\kappa_{stoch}^2}{2}.$$ 
Following the (\ref{eq:final1}) and the derivation leading onto it, we have under $B_1 \cap B_2$, 
\begin{eqnarray}\label{eq:restocheq}
\kappa_{stoch}\|\hat{\bm \beta}_J- \bm \beta_J\|_2 \leq \frac{\sqrt{2}}{\sqrt{n}}\|X(\hat{\bm \beta}-\bm \beta)\|_2.
\end{eqnarray}
which in turn means
\begin{eqnarray}
\vspace{-0.5 in}&&\IP\left(\frac{1}{n}\|X(\hat{\bm \beta}-\bm \beta)\|^2_2 \geq 32 sr^2/\kappa_{stoch}^2\right) 
\leq \IP(B_1^c \cup B_2^c). \nonumber
\end{eqnarray}
Thus the proof is concluded by estimating the probabilities of $B_1^c$ and $B_2^c$ to be small. For $B_1$ the proof consists essentially of same steps as Theorem \ref{th:covariate2}. In particular, note that we no longer have the additional constraint on $X$ that diagonals of $\tran{X}X/n$ is 1 and thus we need to apply the Nagaev type concentration inequalities from Result \ref{th:nagaev nonlinear light} on the mean-zero process $\sum_{i=1}^n x_{i,j}e_i$ directly. Note that, in this scenario, while applying the Nagaev-type inequalities on $S_{n,b}$, we choose $b_k=1$ for $1 \leq k \leq n$ and thus $\|b\|^2_2=n$, $\|b\|_q^q=n$. 

Thus, in the view of the choice of $r$ (refined for the stochastic case) in (\ref{eq:new r}), $\IP(B_1)\to 1$. Lemma \ref{lem:x'x} ensures $\IP(B_2)\to 1$ under the condition on the dependence adjusted functional dependence measure. This in turn yields $\IP(B_1 \cap B_2) \to 1$. This finishes the proof of the Theorem \ref{th:lasso stochastic}.

\end{proof}

\section{Appendix C:- Additional notes on implementation of methods}
\subsection*{Additional notes on implementation
of QTL-LASSO}
We use LASSO implementation in R-package \texttt{glmnet} with tuning parameter $\lambda$ chosen by cross validation and with weights argument $(v_1\ldots,v_n)= ((1-\delta) \delta^{(T-1)}/(1-\delta^T),\ldots,1)$ to account for the structural change in coefficients. $\delta=0.8$.

\subsection*{Additional notes and outputs on implementation of ets, nnar and armax with software output} The
 ETS$(\textrm{A,N,N})$ with tuning parameter $=0.0446$, NNAR$(38,22)$ with one hidden layer and ARMA$(2,1)$ were selected by AIC and estimated by R-package \texttt{forecast}. NNAR  and ARMAX allow for exogenous covariates, therefore we include aggregated weather series $\bar{w}_t=\sum_{k=1}^{73}w_{k,t}$, $\bar{\tau}_t=\sum_{l=1}^{78}\tau_{l,t}$, and weekend   dummies as well. For NNAR, we can provide weights for the covariate observations. We use the same exponential down-weighing scheme as for the QTL-LASSO, but with $\alpha=0.98$, which gave better results. 
 
 First the price series $y_t$ is seasonally adjusted using STL decomposition (R-core function). The seasonally adjusted prices $z_t$ is used as input for the models implemented in R-package \texttt{forecast}. The models are specified as follows:\\

\textbf{ETS} The model is selected according to AIC criterion. We restrict the model without trend component, because the prices do not show any trend pattern (see \ref{fig:prices}). However, probably due to breaks in price-level, the AIC would select a trend component. This results in future paths that vary a lot. For optimization criterion, for, we use Average MSFE, over maximal possible horizon=30 hours. This results into model with tuning parameter  $0.0446$ selected by AIC. This gives better forecasting results than minimizing in-sample MSE which would result in tuning parameter $0.99$ and huge PIs.

\begin{Verbatim}[fontsize=\footnotesize]
ETS (A, N, N)
# means additive model, without 
trend and seasonal components.
Call: ets(y = y, model = "ZNZ", opt.crit = "amse", nmse = 30)

Smoothing parameters:
 delta = 0.0446

Initial states:
 l = 34.1139

 sigma: 8.4748

  AIC     AICc    BIC
116970.9 116970.9 116992.1
\end{Verbatim}

\textbf{NNAR} The model is selected according to AIC criterion. The model is restricted in that it allows only one hidden layer. The number of nodes in this layer is by default given as $(\# \textrm{AR lags} + \# \textrm{exogenous covariates})/2$. In this case, we use aggregated wind speed and temperatures,  and dummies for weekends so the number of exogenous covariates is 4. In order to get fair comparison with the QTL-LASSO, we also use exponential downweighting on the exogenous covariates, this time with tuning parameter $0.95$.

\begin{Verbatim}[fontsize=\footnotesize]
NNAR (38,22) 
# means AR order is 38 and there are 22 nodes 
in the hidden layer
Call:  nnetar(y = y, xreg =  cbind(Weather_agg, dummy_12)
,weights = expWeights(alpha=0.95)))

Average of 20 networks, each of which is a 42-22-1 network with 969 weights
options were - linear output units

sigma^2 estimated as 15.34
\end{Verbatim}

\textbf{ARMA} The model is selected according to AIC criterion.  we use aggregated wind speed and temperatures,  and dummies for the weekend.
\begin{Verbatim}[fontsize=\footnotesize]
Regression with ARIMA(2,0,1) errors

Coefficients:
    ar1   ar2   ma1   intc.   
    0.506 0.328 0.491 59.378 
s.e.0.060 0.055 0.057  1.644  
    xreg1  xreg2 xreg3 xreg4
    -4.551 -0.489 0.481 0.133
s.e.0.404  0.060 0.538 0.538
sigma^2 estimated as 24.35:  
log likelihood=-26410.34
AIC=52838.68   AICc=52838.7   BIC=52902.38
\end{Verbatim}

\section{Appendix D:- Additional simulation results}
\renewcommand{\thesubtable}{\roman{subtable}}
\begin{sidewaystable}
\vspace{1cm}
\begin{subtable}{1\textwidth}
\scalebox{0.6}{
 \input{Tables/TAB_SIM_90_LD_uniform.tex}
 }
\caption{\footnotesize{Case $n>p$. QTL implemented using each of the estimators OLS, LAD or LASSO and CLT using LASSO only.\vspace{0.1cm}}}\label{tab:SIM_90_LD_uniform }
\end{subtable}
\begin{subtable}{1\textwidth}
\scalebox{0.6}{
 \input{Tables/TAB_SIM_90_HD_uniform.tex}
 }
\caption{ \footnotesize{Case $p>n$. QTL and CLT as above. ADJ is a bootstrap version QTL for better performance under short-sample.}}\label{tab:SIM_90_HD_uniform }
\end{subtable}
\caption{\footnotesize{Simulated out-of-sample forecasting experiment. The reported values are coverage probabilities, i.e., relative (\%) counts of out-of-sample values covered in 1000 trials (left part) and average Winkler loss values (right part). The nominal coverage (the first column) ranks from is $95\%$ to $60\%$. Simulated error processes heavy tails and either short memory, long memory or are nonlinear. The elements of regression coefficient $\bm \beta$ are drawn independently from uniform distribution $U[-1,1]$. The sparsity of $\bm \beta$'s is fixed to 90\%.} }\label{tab:SIM_90_uniform }
\end{sidewaystable}

\renewcommand{\thesubtable}{\roman{subtable}}
\begin{sidewaystable}
\vspace{1cm}
\begin{subtable}{1\textwidth}
\scalebox{0.6}{
 \input{Tables/TAB_SIM_20_LD_uniform.tex}
 }
\caption{\footnotesize{Case $n>p$. QTL implemented using each of the estimators OLS, LAD or LASSO and CLT using LASSO only.\vspace{0.1cm}}}\label{tab:SIM_20_LD_uniform }
\end{subtable}
\begin{subtable}{1\textwidth}
\scalebox{0.6}{
 \input{Tables/TAB_SIM_20_HD_uniform.tex}
 }
\caption{ \footnotesize{Case $p>n$. QTL and CLT as above. ADJ is a bootstrap version QTL for better performance under short-sample.}}\label{tab:SIM_20_HD_uniform }
\end{subtable}
\caption{\footnotesize{Simulated out-of-sample forecasting experiment. The reported values are coverage probabilities, i.e., relative (\%) counts of out-of-sample values covered in 1000 trials (left part) and average Winkler loss values (right part). The nominal coverage (the first column) ranks from is $95\%$ to $60\%$. Simulated error processes heavy tails and either short memory, long memory or are nonlinear. The elements of regression coefficient $\bm \beta$ are drawn independently from uniform distribution $U[-1,1]$. The sparsity of $\bm \beta$'s is fixed to 20\% }. }\label{tab:SIM_20_uniform }
\end{sidewaystable}

\end{document}

%% file: Tables/TAB_SIM_50_LD_uniform.tex
\begin{tabular*}{\textwidth}{@{\extracolsep{\fill}} |llrrrr|rrrr|rrrr|R{0.8cm}R{0.8cm}R{0.8cm}R{0.8cm}|rrrr|R{0.8cm}R{0.8cm}R{0.8cm}R{0.6cm}|}
\cmidrule{1-26}
\multirow{3}{*}{\rotatebox[origin=c]{90}{nominal}}& & \multicolumn{12}{c|}{Coverage}& \multicolumn{12}{|c|}{Winkler loss}\\
&$(e_i)$& \multicolumn{4}{c}{short-heavy}& \multicolumn{4}{c}{long-heavy}&\multicolumn{4}{c|}{non-lin-heavy}& \multicolumn{4}{|c}{short-heavy}& \multicolumn{4}{c}{long-heavy}&\multicolumn{4}{c|}{non-lin-heavy}\\
&$m$-weeks& 1 & 2 & 3 & 4 & 1 & 2 & 3 & 4  & 1 & 2 & 3 & 4  & 1 & 2 & 3 & 4 & 1 & 2 & 3 & 4 & 1 & 2 & 3 & 4 \\
\cmidrule{2-26}
\multirow{2}{*}{\rotatebox[origin=c]{90}{$60\%$}}
&ols-qtl & 47.0 & 41.0 & 38.0 & 37.3   & 38.9 & 33.1 & 28.8 & 25.9   & 48.6 & 42.0 & 41.0 & 37.9   & 2.1 & 1.7 & 1.5 & 1.4   & 7.0 & 6.4 & 6.5 & 6.6   & 1.5 & 1.2 & 1.1 & 1.0 \\
&lad-qtl & 51.5 & \textbf{51.1} & 48.5 & \textbf{45.9}   & 38.0 & 33.4 & 32.1 & 29.2   & 50.7 & \textbf{51.8} & \textbf{49.1} & \textbf{50.0}   & \textbf{1.5} & \textbf{1.2} & \textbf{1.1} & \textbf{1.0}   & 6.0 & 5.5 & 5.6 & 5.6   & \textbf{1.0} & \textbf{0.8}& \textbf{0.8} & \textbf{0.7} \\
&lss-qtl & \textbf{53.4} & 50.3 & \textbf{49.8} & 36.7   & \textbf{55.4} & \textbf{53.0} & \textbf{54.6} & \textbf{48.8}   & \textbf{55.0} & 49.5 & 42.6 & 43.1   & 1.6 & 1.4 & 1.2 & 1.7   & \textbf{6.0} & \textbf{5.2} & \textbf{5.2} & \textbf{5.1}   & 1.2 & 1.0 & 1.2 & 1.1 \\
&lss-clt & 48.2 & 35.9 & 35.2 & 19.3   & 43.9 & 32.3 & 28.4 & 23.7   & 48.8 & 33.7 & 23.0 & 23.1   & 1.7 & 1.6 & 1.4 & 2.1   & 6.1 & 5.6 & 5.6 & 5.7   & 1.3 & 1.2 & 1.4 & 1.3 \\[.2cm]
\multirow{2}{*}{\rotatebox[origin=c]{90}{$80\%$}}
&ols-qtl & 65.9 & 60.5 & 55.6 & 52.2   & 56.5 & 50.1 & 44.6 & 39.3   & 67.4 & 62.5 & 59.5 & 54.2   & 3.0 & 2.5 & 2.2 & 2.0   & 10.3 & 9.7 & 10.0 & 10.4  & 2.2 & 1.8 & 1.6 & 1.5 \\
&lad-qtl & 74.5 & 71.5 & 71.0 & \textbf{68.7}   & 61.4 & 53.9 & 53.2 & 49.0   & 73.7 & 72.3 & \textbf{72.4} & \textbf{69.7}   & \textbf{2.1} & \textbf{1.7} & \textbf{1.6} & \textbf{1.5}   & 8.8 & 8.0 & 8.6 & 8.6     & \textbf{1.5} & \textbf{1.2} & \textbf{1.1} & \textbf{1.1} \\
&lss-qtl & \textbf{77.0} & \textbf{74.6} & \textbf{73.5} & 60.4   & \textbf{75.1} & \textbf{74.9} & \textbf{73.0} & \textbf{70.2}   & \textbf{75.3} &\textbf{74.3} & 68.3 & 66.6   & 2.3 & 1.9 & 1.8 & 2.2   & \textbf{8.5} & \textbf{7.4} & \textbf{7.7} & \textbf{7.6}     & 1.6 & 1.4 & 1.5 & 1.4 \\
&lss-clt & 70.1 & 54.1 & 51.4 & 29.3   & 60.5 & 47.5 & 43.7 & 35.4   & 65.1 & 51.0 & 36.6 & 35.9   & 2.4 & 2.2 & 2.0 & 3.2   & 8.9 & 8.3 & 8.7 & 9.0     & 1.8 & 1.7 & 2.1 & 1.9 \\[.2cm]
\multirow{2}{*}{\rotatebox[origin=c]{90}{$90\%$}}
&ols-qtl & 78.4 & 73.8 & 70.2 & 65.6   & 70.1 & 63.7 & 60.8 & 52.6   & 78.7 & 74.4 & 71.9 & 68.1   & 4.3 & 3.6 & 3.3 & 3.2   & 14.6 & 14.1 & 15.2 & 16.2   & 2.9 & 2.6 & 2.4 & 2.3 \\
&lad-qtl & 87.3 & 84.1 & 84.1 & \textbf{83.2}   & 76.9 & 73.0 & 71.9 & 65.6   & 87.4 & 83.9 & \textbf{85.8} & 81.5   & \textbf{3.1} & \textbf{2.5} & 2.9 & \textbf{2.5}   & 12.2 & 11.5 & 13.4 & 13.5   & \textbf{2.1} & \textbf{1.8} & \textbf{2.1} & \textbf{1.9} \\
&lss-qtl & \textbf{87.7} & \textbf{87.5} & \textbf{85.7} & 75.2   & \textbf{86.5} & \textbf{86.1} & \textbf{83.5} & \textbf{81.6}   & \textbf{88.3} & \textbf{85.4} & 82.9 & \textbf{81.9}   & 3.2 & 2.6 & 3.0 & 3.2   & \textbf{11.7} & \textbf{10.8} & \textbf{11.8} & \textbf{11.6}   & 2.2 & 1.9 & 2.4 & 2.2 \\
&lss-clt & 79.7 & 68.2 & 63.8 & 39.2   & 71.0 & 59.2 & 54.5 & 46.8   & 75.8 & 64.5 & 45.7 & 45.4   & 3.3 & 3.1 & 2.8 & 5.0   & 12.7 & 12.4 & 13.6 & 14.2   & 2.5 & 2.4 & 3.2 & 3.0 \\[.2cm]
\multirow{2}{*}{\rotatebox[origin=c]{90}{$95\%$}}
&ols-qtl & 86.3 & 82.4 & 76.4 & 72.1   & 81.9 & 75.1 & 69.1 & 60.7   & 86.8 & 82.2 & 78.6 & 73.3   & 6.1 & 5.5 & 5.0 & 4.9   & 20.5 & 21.2 & 24.1 & 26.2   & 3.9 & 4.1 & 3.7 & 3.6 \\
&lad-qtl & 92.7 & 91.2 & 88.1 & \textbf{87.1}   & 87.0 & 83.6 & 77.2 & 71.4   & 92.7 & 90.9 & \textbf{89.5} & \textbf{86.4}   & \textbf{4.5} & \textbf{4.5} & \textbf{3.9} & \textbf{3.5}   & 16.8 & 17.7 & 20.0 & 20.4   & \textbf{3.1} & \textbf{3.1} & \textbf{2.9} & \textbf{2.7} \\
&lss-qtl & \textbf{93.5} & \textbf{92.3} & \textbf{90.1} & 80.9   & \textbf{92.6} & \textbf{90.6} & \textbf{87.6} & \textbf{86.8}   & \textbf{92.8} & \textbf{91.7} & 87.8 & \textbf{86.4}   & 4.6 & 4.6 & 4.0 & 4.3   & \textbf{16.5} & \textbf{16.6} & \textbf{16.8} & \textbf{16.2}   & \textbf{3.1} & \textbf{3.1} & 3.2 & 3.0 \\
&lss-clt & 87.5 & 76.6 & 73.8 & 46.3   & 78.7 & 67.0 & 62.2 & 53.7   & 82.4 & 71.5 & 54.5 & 55.6   & 4.8 & 4.4 & 4.2 & 7.9   & 18.7 & 18.9 & 22.0 & 23.0   & 3.4 & 3.4 & 4.9 & 4.7 \\[.2cm]
\cmidrule{1-26}
\end{tabular*}

%% file: Tables/TAB_SIM_50_HD_uniform.tex
\begin{tabular*}{\textwidth}{@{\extracolsep{\fill}} |llrrrr|rrrr|rrrr|rrrr|rrrr|rrrr|}
\cmidrule{1-26}
\multirow{3}{*}{\rotatebox[origin=c]{90}{nominal}}&$\beta$& \multicolumn{12}{c|}{Coverage}& \multicolumn{12}{|c|}{Winkler loss}\\
&$(e_i)$& \multicolumn{4}{c}{short-heavy}& \multicolumn{4}{c}{long-heavy}&\multicolumn{4}{c|}{non-lin-heavy}& \multicolumn{4}{|c}{short-heavy}& \multicolumn{4}{c}{long-heavy}&\multicolumn{4}{c|}{non-lin-heavy}\\
&$m$-days& 1 & 2 & 3 & 4 & 1 & 2 & 3 & 4  & 1 & 2 & 3 & 4  & 1 & 2 & 3 & 4 & 1 & 2 & 3 & 4 & 1 & 2 & 3 & 4 \\
\cmidrule{2-26}
\multirow{3}{*}{\rotatebox[origin=c]{90}{$60\%$}}
&lss-qtl & 50.5 & \textbf{49.4} & 47.0 & 41.3   & 52.7 & 47.1 & \textbf{43.9} & 38.7   & \textbf{51.3} & 47.7 & 47.2 & 40.2   & \textbf{4.8} & \textbf{3.8} & 3.8 & 3.4   & 8.6 & 7.7 & 7.8 & 7.8     & \textbf{3.1} & \textbf{2.5} & 2.5 & 2.3 \\
&lss-clt & \textbf{54.6} & 39.6 & 37.2 & 32.2   & \textbf{68.4} & \textbf{49.5} & 42.2 & 33.8   & 48.3 & 35.3 & 33.6 & 27.8   & 5.5 & 4.9 & 4.9 & 4.3   & \textbf{8.0} & \textbf{7.1} & 7.4 & 7.5     & 3.9 & 3.6 & 3.6 & 3.5 \\
&lss-adj & 50.2 & 48.0 & \textbf{48.0}& \textbf{47.9}   & 50.9 & 46.5 & 43.0 & \textbf{40.0}   & 51.1 & \textbf{48.6} & \textbf{48.9} & \textbf{48.6}   & 4.9 & 3.9 & \textbf{3.7} & \textbf{3.3}   & 8.4 & 7.2 & \textbf{6.9} & \textbf{6.8}   & \textbf{3.1} & \textbf{2.5} & \textbf{2.4} & \textbf{2.1} \\[.2cm]
\multirow{3}{*}{\rotatebox[origin=c]{90}{$80\%$}}
&lss-qtl & 71.8 & \textbf{66.6} & 61.2 & 54.1   & 70.0 & 64.2 & 56.6 & 50.0   & 71.2 & 65.7 & 62.4 & 55.0   & \textbf{7.4} & 6.2 & \textbf{5.4} & 5.1   & 13.5 & 12.0 & 11.8 & 11.7   & \textbf{4.6} & \textbf{4.0} & 3.6 & 3.3 \\
&lss-clt & \textbf{73.9} & 56.8 & 51.7 & 48.9   & \textbf{81.6} & \textbf{68.0} & \textbf{59.3} & 51.2   & 70.3 & 55.7 & 47.8 & 40.7   & 7.8 & 6.7 & 7.1 & 6.3   & \textbf{12.5} & \textbf{10.5} & 10.8 & 11.1   & 5.3 & 5.0 & 5.2 & 5.3 \\
&lss-adj & 71.4 & 65.5 & \textbf{63.5} & \textbf{62.6}   & 68.9 & 63.7 & 58.0 & \textbf{54.7}   & \textbf{72.0} & \textbf{67.7} & \textbf{67.4} & \textbf{64.9}   & 7.5 & \textbf{6.1} & 5.6 & \textbf{4.9}   & 12.6 & 10.9 & \textbf{10.7} & \textbf{10.5}   & 4.8 & \textbf{4.0} & \textbf{3.5} & \textbf{3.1} \\[.2cm]
\multirow{3}{*}{\rotatebox[origin=c]{90}{$90\%$}}
&lss-qtl & 82.9 & 74.2 & 67.7 & 60.6   & 80.6 & 73.1 & 62.5 & 55.2   & \textbf{82.1} & 73.9 & 68.1 & 62.0   & 12.4 & 9.0 & 8.0 & 7.9  & 20.7 & 17.4 & 17.9 & 18.7   & 7.9 & 5.9 & 5.3 & 5.0 \\
&lss-clt & \textbf{84.0} & 69.7 & 63.5 & 59.9   & \textbf{88.3} & \textbf{78.2} & \textbf{70.4} & 62.9   & 81.1 & 68.7 & 59.6 & 52.3   & \textbf{11.5} &\textbf{8.8} & 9.9 & 8.9  & \textbf{18.8} & \textbf{15.7} & \textbf{16.1} & 16.7   & 7.4 & 6.7 & 7.4 & 7.8 \\
&lss-adj & 80.8 & \textbf{75.8} & \textbf{74.2} & \textbf{72.5}   & 79.3 & 73.7 & 67.6 & \textbf{64.5}   & 81.0 & \textbf{77.9} & \textbf{76.2} & \textbf{73.9}   & 12.0 & \textbf{8.8} & \textbf{8.0} & \textbf{7.3}  & 19.5 & 16.5 & 16.4 & \textbf{16.5}   & \textbf{7.3} & \textbf{5.7} & \textbf{5.0} & \textbf{4.5} \\[.2cm]
\multirow{3}{*}{\rotatebox[origin=c]{90}{$95\%$}}
&lss-qtl & 86.2 & 77.7 & 72.4 & 64.6   & 85.1 & 76.3 & 66.6 & 57.7   & 85.2 & 78.1 & 72.5 & 66.6   & 19.3 & 14.0 & 12.9 & 13.2   & 32.9 & 27.9 & 30.0 & 32.0   & 11.3 & 9.1 & 8.4 & 8.2 \\
&lss-clt & \textbf{90.1} & 80.1 & 74.0 & 70.4   & \textbf{91.8} & \textbf{83.5} & \textbf{76.9} & \textbf{71.3}   & 85.8 & 77.8 & 70.6 & 62.6   & \textbf{17.4} & \textbf{12.4} & 13.6 & 12.5   & \textbf{29.2} & \textbf{24.2} & \textbf{25.0} & \textbf{26.0}   & \textbf{10.5} & 9.0 & 10.4 & 11.5 \\
&lss-adj & 86.3 & \textbf{81.9} & \textbf{79.8} & \textbf{78.1}   & 85.8 & 79.2 & 74.5 & 68.5   & \textbf{86.0} & \textbf{84.2} & \textbf{81.6} & \textbf{80.7}   & 18.7 & 13.0 & \textbf{11.8} & \textbf{11.3}   & 30.8 & 25.9 & 26.3 & 26.8   & 10.9 & \textbf{8.4} & \textbf{7.5} & \textbf{6.8} \\[.2cm]
\cmidrule{1-26}
\end{tabular*}

%% file: Tables/TAB_SIM_90_LD_uniform.tex
\begin{tabular*}{\textwidth}{@{\extracolsep{\fill}} |llrrrr|rrrr|rrrr|R{0.7cm}R{0.7cm}R{0.7cm}R{0.7cm}|rrrr|R{0.7cm}R{0.7cm}R{0.7cm}R{0.7cm}|}
\cmidrule{1-26}
\multirow{3}{*}{\rotatebox[origin=c]{90}{nominal}}& & \multicolumn{12}{c|}{Coverage}& \multicolumn{12}{|c|}{Winkler loss}\\
&$(e_i)$& \multicolumn{4}{c}{short-heavy}& \multicolumn{4}{c}{long-heavy}&\multicolumn{4}{c|}{non-lin-heavy}& \multicolumn{4}{|c}{short-heavy}& \multicolumn{4}{c}{long-heavy}&\multicolumn{4}{c|}{non-lin-heavy}\\
&$m$-weeks& 1 & 2 & 3 & 4 & 1 & 2 & 3 & 4  & 1 & 2 & 3 & 4  & 1 & 2 & 3 & 4 & 1 & 2 & 3 & 4 & 1 & 2 & 3 & 4 \\
\cmidrule{2-26}
\multirow{2}{*}{\rotatebox[origin=c]{90}{$60\%$}}
&ols-qtl & 47.0 & 41.0 & 38.0 & 37.3   & 38.9 & 33.1 & 28.8 & 25.9   & 48.6 & 42.0 & 41.0 & 37.9   & 2.1 & 1.7 & 1.5 & 1.4   & 7.0 & 6.4 & 6.5 & 6.6   & 1.5 & 1.2 & 1.1 & 1.0 \\
&lad-qtl & 51.5 & 51.1 & 48.5 & 45.9   & 38.0 & 33.4 & 32.1 & 29.2   & 50.7 & 51.8 & 49.1 & 50.0   & 1.5 & 1.2 & 1.1 & 1.0   & 6.0 & 5.5 & 5.6 & 5.6   & 1.0 & 0.8 & 0.8 & 0.7 \\
&lss-qtl & 56.8 & 55.6 & 54.2 & 49.7   & 54.9 & 53.6 & 53.0 & 50.0   & 57.4 & 55.8 & 51.8 & 49.6   & 1.5 & 1.2 & 1.1 & 1.1   & 5.9 & 5.2 & 5.2 & 5.1   & 1.1 & 0.9 & 0.9 & 0.8 \\
&lss-clt & 54.3 & 46.5 & 43.1 & 37.5   & 43.3 & 30.8 & 28.0 & 23.6   & 51.8 & 42.4 & 35.0 & 33.2   & 1.5 & 1.3 & 1.2 & 1.2   & 6.0 & 5.5 & 5.6 & 5.6   & 1.1 & 1.0 & 0.9 & 0.9 \\[.2cm]
\multirow{2}{*}{\rotatebox[origin=c]{90}{$80\%$}}
&ols-qtl & 65.9 & 60.5 & 55.6 & 52.2   & 56.5 & 50.1 & 44.6 & 39.3   & 67.4 & 62.5 & 59.5 & 54.2   & 3.0 & 2.5 & 2.2 & 2.0   & 10.3 & 9.7 & 10.0 & 10.4  & 2.2 & 1.8 & 1.6 & 1.5 \\
&lad-qtl & 74.5 & 71.5 & 71.0 & 68.7   & 61.4 & 53.9 & 53.2 & 49.0   & 73.7 & 72.3 & 72.4 & 69.7   & 2.1 & 1.7 & 1.6 & 1.5   & 8.8 & 8.0 & 8.6 & 8.6     & 1.5 & 1.2 & 1.1 & 1.1 \\
&lss-qtl & 80.1 & 77.9 & 74.1 & 71.2   & 74.8 & 75.4 & 73.5 & 70.8   & 77.2 & 75.4 & 73.7 & 71.4   & 2.1 & 1.7 & 1.6 & 1.6   & 8.4 & 7.3 & 7.7 & 7.6     & 1.5 & 1.3 & 1.2 & 1.2 \\
&lss-clt & 74.1 & 65.8 & 62.1 & 53.2   & 62.3 & 49.9 & 42.7 & 36.8   & 71.1 & 61.3 & 54.1 & 50.6   & 2.2 & 1.7 & 1.7 & 1.8   & 8.7 & 8.2 & 8.7 & 8.8     & 1.6 & 1.4 & 1.4 & 1.3 \\[.2cm]
\multirow{2}{*}{\rotatebox[origin=c]{90}{$90\%$}}
&ols-qtl & 78.4 & 73.8 & 70.2 & 65.6   & 70.1 & 63.7 & 60.8 & 52.6   & 78.7 & 74.4 & 71.9 & 68.1   & 4.3 & 3.6 & 3.3 & 3.2   & 14.6 & 14.1 & 15.2 & 16.2   & 2.9 & 2.6 & 2.4 & 2.3 \\
&lad-qtl & 87.3 & 84.1 & 84.1 & 83.2   & 76.9 & 73.0 & 71.9 & 65.6   & 87.4 & 83.9 & 85.8 & 81.5   & 3.1 & 2.5 & 2.9 & 2.5   & 12.2 & 11.5 & 13.4 & 13.5   & 2.1 & 1.8 & 2.1 & 1.9 \\
&lss-qtl & 88.9 & 87.8 & 87.0 & 83.9   & 87.5 & 86.4 & 83.1 & 81.2   & 88.9 & 85.9 & 85.9 & 83.1   & 3.1 & 2.5 & 2.8 & 2.6   & 11.7 & 10.8 & 11.8 & 11.6   & 2.1 & 1.8 & 2.1 & 1.9 \\
&lss-clt & 85.1 & 78.7 & 73.7 & 64.7   & 72.1 & 59.4 & 54.5 & 47.5   & 81.3 & 72.7 & 65.4 & 63.4   & 3.1 & 2.5 & 2.5 & 2.6   & 12.5 & 12.2 & 13.5 & 14.0   & 2.2 & 1.9 & 2.0 & 2.0 \\[.2cm]
\multirow{2}{*}{\rotatebox[origin=c]{90}{$95\%$}}
&ols-qtl & 86.3 & 82.4 & 76.4 & 72.1   & 81.9 & 75.1 & 69.1 & 60.7   & 86.8 & 82.2 & 78.6 & 73.3   & 6.1 & 5.5 & 5.0 & 4.9   & 20.5 & 21.2 & 24.1 & 26.2   & 3.9 & 4.1 & 3.7 & 3.6 \\
&lad-qtl & 92.7 & 91.2 & 88.1 & 87.1   & 87.0 & 83.6 & 77.2 & 71.4   & 92.7 & 90.9 & 89.5 & 86.4   & 4.5 & 4.5 & 3.9 & 3.5   & 16.8 & 17.7 & 20.0 & 20.4   & 3.1 & 3.1 & 2.9 & 2.7 \\
&lss-qtl & 93.9 & 93.3 & 90.0 & 87.5   & 92.4 & 91.3 & 88.2 & 86.7   & 93.5 & 92.5 & 90.7 & 88.0   & 4.5 & 4.4 & 3.8 & 3.5   & 16.5 & 16.5 & 16.7 & 16.3   & 3.0 & 3.1 & 2.9 & 2.7 \\
&lss-clt & 89.0 & 85.7 & 80.0 & 74.1   & 78.4 & 67.6 & 62.1 & 55.1   & 87.0 & 79.5 & 75.0 & 72.2   & 4.5 & 3.6 & 3.7 & 3.8   & 18.5 & 18.6 & 21.8 & 22.8   & 3.1 & 2.8 & 3.0 & 3.0 \\[.2cm]
\cmidrule{1-26}
\end{tabular*}

%% file: Tables/TAB_SIM_90_HD_uniform.tex
\begin{tabular*}{\textwidth}{@{\extracolsep{\fill}} |llrrrr|rrrr|rrrr|rrrr|rrrr|rrrr|}
\cmidrule{1-26}
\multirow{3}{*}{\rotatebox[origin=c]{90}{nominal}}&$\beta$& \multicolumn{12}{c|}{Coverage}& \multicolumn{12}{|c|}{Winkler loss}\\
&$(e_i)$& \multicolumn{4}{c}{short-heavy}& \multicolumn{4}{c}{long-heavy}&\multicolumn{4}{c|}{non-lin-heavy}& \multicolumn{4}{|c}{short-heavy}& \multicolumn{4}{c}{long-heavy}&\multicolumn{4}{c|}{non-lin-heavy}\\
&$m$-days& 1 & 2 & 3 & 4 & 1 & 2 & 3 & 4  & 1 & 2 & 3 & 4  & 1 & 2 & 3 & 4 & 1 & 2 & 3 & 4 & 1 & 2 & 3 & 4 \\
\cmidrule{2-26}
\multirow{3}{*}{\rotatebox[origin=c]{90}{$60\%$}}
&lss-qtl & 50.9 & 48.4 & 47.2 & 41.7   & 53.0 & 46.6 & 44.0 & 38.6   & 49.9 & 48.5 & 46.3 & 39.4   & 4.8 & 3.8 & 3.8 & 3.5   & 8.6 & 7.7 & 7.8 & 7.8   & 3.1 & 2.5 & 2.5 & 2.2 \\
&lss-clt & 54.6 & 39.3 & 37.4 & 32.7   & 68.8 & 48.9 & 41.8 & 33.4   & 48.2 & 33.9 & 34.0 & 27.7   & 5.5 & 4.9 & 4.9 & 4.3   & 8.0 & 7.1 & 7.4 & 7.5   & 3.9 & 3.6 & 3.6 & 3.6 \\
&lss-adj & 50.8 & 47.5 & 48.1 & 47.1   & 51.3 & 46.9 & 43.4 & 40.3   & 51.2 & 48.8 & 49.0 & 48.6   & 4.8 & 3.9 & 3.7 & 3.3   & 8.4 & 7.2 & 7.0 & 6.8   & 3.1 & 2.5 & 2.3 & 2.1 \\[.2cm]
\multirow{3}{*}{\rotatebox[origin=c]{90}{$80\%$}}
&lss-qtl & 71.7 & 66.3 & 61.2 & 53.4   & 70.5 & 64.3 & 56.0 & 50.6   & 69.7 & 64.8 & 62.2 & 54.4   & 7.4 & 6.2 & 5.4 & 5.1   & 13.5 & 11.9 & 11.9 & 11.7   & 4.6 & 4.1 & 3.6 & 3.3 \\
&lss-clt & 73.4 & 57.6 & 51.8 & 49.1   & 82.2 & 68.2 & 59.7 & 51.3   & 69.6 & 54.4 & 47.7 & 40.5   & 7.9 & 6.7 & 7.1 & 6.3   & 12.5 & 10.5 & 10.8 & 11.2   & 5.3 & 5.0 & 5.2 & 5.3 \\
&lss-adj & 71.7 & 65.1 & 65.0 & 63.0   & 69.2 & 64.7 & 57.9 & 54.5   & 71.6 & 67.1 & 66.9 & 65.3   & 7.5 & 6.1 & 5.6 & 5.0   & 12.5 & 10.8 & 10.7 & 10.5   & 4.9 & 4.0 & 3.5 & 3.1 \\[.2cm]
\multirow{3}{*}{\rotatebox[origin=c]{90}{$90\%$}}
&lss-qtl & 82.7 & 74.0 & 68.6 & 61.3   & 81.2 & 73.1 & 62.5 & 55.2   & 81.8 & 73.4 & 67.9 & 62.2   & 12.5 & 9.0 & 8.1 & 7.9  & 20.7 & 17.4 & 18.0 & 18.7   & 7.9 & 6.0 & 5.3 & 5.0 \\
&lss-clt & 83.7 & 69.3 & 63.6 & 59.3   & 88.9 & 78.0 & 70.2 & 62.9   & 80.9 & 68.7 & 58.8 & 52.4   & 11.5 & 8.8 & 9.8 & 8.9  & 18.8 & 15.7 & 16.0 & 16.7   & 7.4 & 6.7 & 7.3 & 7.8 \\
&lss-adj & 81.4 & 75.3 & 73.9 & 72.1   & 79.3 & 73.3 & 67.7 & 63.8   & 81.5 & 77.8 & 74.9 & 74.4   & 12.0 & 8.8 & 8.1 & 7.4  & 19.4 & 16.5 & 16.4 & 16.5   & 7.5 & 5.7 & 5.0 & 4.5 \\[.2cm]
\multirow{3}{*}{\rotatebox[origin=c]{90}{$95\%$}}
&lss-qtl & 86.1 & 77.6 & 73.3 & 65.8   & 85.3 & 76.8 & 66.4 & 57.2   & 85.5 & 77.9 & 71.7 & 66.0   & 19.4 & 13.9 & 13.0 & 13.3   & 32.9 & 27.9 & 30.0 & 32.0   & 11.3 & 9.2 & 8.4 & 8.3 \\
&lss-clt & 90.3 & 80.0 & 74.3 & 70.4   & 91.6 & 83.7 & 76.5 & 71.3   & 86.1 & 77.9 & 70.6 & 62.5   & 17.5 & 12.4 & 13.6 & 12.5   & 29.2 & 24.1 & 24.9 & 26.0   & 10.6 & 9.1 & 10.2 & 11.5 \\
&lss-adj & 86.2 & 82.1 & 80.6 & 78.4   & 85.8 & 79.3 & 75.0 & 68.3   & 86.4 & 83.0 & 81.7 & 79.6   & 18.6 & 13.0 & 11.9 & 11.4   & 31.1 & 25.8 & 26.3 & 26.9   & 11.0 & 8.5 & 7.4 & 6.9 \\[.2cm]
\cmidrule{1-26}
\end{tabular*}

%% file: Tables/TAB_SIM_20_LD_uniform.tex
\begin{tabular*}{\textwidth}{@{\extracolsep{\fill}} |llrrrr|rrrr|rrrr|R{0.7cm}R{0.7cm}R{0.7cm}R{0.7cm}|rrrr|R{0.7cm}R{0.7cm}R{0.7cm}R{0.7cm}|}
\cmidrule{1-26}
\multirow{3}{*}{\rotatebox[origin=c]{90}{nominal}}& & \multicolumn{12}{c|}{Coverage}& \multicolumn{12}{|c|}{Winkler loss}\\
&$(e_i)$& \multicolumn{4}{c}{short-heavy}& \multicolumn{4}{c}{long-heavy}&\multicolumn{4}{c|}{non-lin-heavy}& \multicolumn{4}{|c}{short-heavy}& \multicolumn{4}{c}{long-heavy}&\multicolumn{4}{c|}{non-lin-heavy}\\
&$m$-weeks& 1 & 2 & 3 & 4 & 1 & 2 & 3 & 4  & 1 & 2 & 3 & 4  & 1 & 2 & 3 & 4 & 1 & 2 & 3 & 4 & 1 & 2 & 3 & 4 \\
\cmidrule{2-26}
\multirow{2}{*}{\rotatebox[origin=c]{90}{$60\%$}}
&ols-qtl & 47.0 & 41.0 & 38.0 & 37.3   & 38.9 & 33.1 & 28.8 & 25.9   & 48.6 & 42.0 & 41.0 & 37.9   & 2.1 & 1.7 & 1.5 & 1.4   & 7.0 & 6.4 & 6.5 & 6.6   & 1.5 & 1.2 & 1.1 & 1.0 \\
&lad-qtl & 51.5 & 51.1 & 48.5 & 45.9   & 38.0 & 33.4 & 32.1 & 29.2   & 50.7 & 51.8 & 49.1 & 50.0   & 1.5 & 1.2 & 1.1 & 1.0   & 6.0 & 5.5 & 5.6 & 5.6   & 1.0 & 0.8 & 0.8 & 0.7 \\
&lss-qtl & 54.8 & 47.5 & 49.9 & 28.2   & 54.7 & 52.5 & 53.4 & 50.3   & 55.7 & 52.1 & 40.1 & 38.8   & 1.7 & 1.6 & 1.3 & 2.2   & 6.2 & 5.3 & 5.3 & 5.1   & 1.1 & 1.0 & 1.3 & 1.2 \\
&lss-clt & 50.1 & 28.8 & 34.6 & 9.9    & 43.0 & 32.0 & 26.1 & 22.8   & 49.6 & 36.5 & 17.3 & 17.7   & 1.8 & 1.8 & 1.5 & 2.7   & 6.4 & 5.7 & 5.7 & 5.6   & 1.2 & 1.2 & 1.6 & 1.5 \\[.2cm]
\multirow{2}{*}{\rotatebox[origin=c]{90}{$80\%$}}
&ols-qtl & 65.9 & 60.5 & 55.6 & 52.2   & 56.5 & 50.1 & 44.6 & 39.3   & 67.4 & 62.5 & 59.5 & 54.2   & 3.0 & 2.5 & 2.2 & 2.0   & 10.3 & 9.7 & 10.0 & 10.4  & 2.2 & 1.8 & 1.6 & 1.5 \\
&lad-qtl & 74.5 & 71.5 & 71.0 & 68.7   & 61.4 & 53.9 & 53.2 & 49.0   & 73.7 & 72.3 & 72.4 & 69.7   & 2.1 & 1.7 & 1.6 & 1.5   & 8.8 & 8.0 & 8.6 & 8.6     & 1.5 & 1.2 & 1.1 & 1.1 \\
&lss-qtl & 75.6 & 71.2 & 72.2 & 55.0   & 74.4 & 73.8 & 74.0 & 70.4   & 77.9 & 74.0 & 66.3 & 65.5   & 2.3 & 2.1 & 1.8 & 2.6   & 8.7 & 7.6 & 7.7 & 7.6     & 1.6 & 1.4 & 1.7 & 1.5 \\
&lss-clt & 65.7 & 46.9 & 50.1 & 16.6   & 59.0 & 46.6 & 40.9 & 36.8   & 70.3 & 53.9 & 28.8 & 26.9   & 2.5 & 2.6 & 2.1 & 4.3   & 9.2 & 8.5 & 8.9 & 8.9     & 1.6 & 1.7 & 2.5 & 2.4 \\[.2cm]
\multirow{2}{*}{\rotatebox[origin=c]{90}{$90\%$}}
&ols-qtl & 78.4 & 73.8 & 70.2 & 65.6   & 70.1 & 63.7 & 60.8 & 52.6   & 78.7 & 74.4 & 71.9 & 68.1   & 4.3 & 3.6 & 3.3 & 3.2   & 14.6 & 14.1 & 15.2 & 16.2   & 2.9 & 2.6 & 2.4 & 2.3 \\
&lad-qtl & 87.3 & 84.1 & 84.1 & 83.2   & 76.9 & 73.0 & 71.9 & 65.6   & 87.4 & 83.9 & 85.8 & 81.5   & 3.1 & 2.5 & 2.9 & 2.5   & 12.2 & 11.5 & 13.4 & 13.5   & 2.1 & 1.8 & 2.1 & 1.9 \\
&lss-qtl & 88.4 & 85.2 & 85.6 & 72.3   & 86.2 & 85.0 & 83.9 & 81.6   & 89.2 & 86.4 & 80.4 & 80.3   & 3.3 & 2.9 & 3.0 & 3.5   & 11.9 & 10.9 & 11.8 & 11.7   & 2.2 & 1.9 & 2.5 & 2.3 \\
&lss-clt & 78.2 & 60.5 & 61.9 & 23.2   & 67.7 & 58.0 & 52.2 & 47.7   & 80.4 & 64.6 & 38.9 & 36.6   & 3.4 & 3.6 & 3.1 & 6.9   & 13.2 & 12.7 & 13.8 & 14.0   & 2.3 & 2.4 & 3.9 & 3.7 \\[.2cm]
\multirow{2}{*}{\rotatebox[origin=c]{90}{$95\%$}}
&ols-qtl & 86.3 & 82.4 & 76.4 & 72.1   & 81.9 & 75.1 & 69.1 & 60.7   & 86.8 & 82.2 & 78.6 & 73.3   & 6.1 & 5.5 & 5.0 & 4.9   & 20.5 & 21.2 & 24.1 & 26.2   & 3.9 & 4.1 & 3.7 & 3.6 \\
&lad-qtl & 92.7 & 91.2 & 88.1 & 87.1   & 87.0 & 83.6 & 77.2 & 71.4   & 92.7 & 90.9 & 89.5 & 86.4   & 4.5 & 4.5 & 3.9 & 3.5   & 16.8 & 17.7 & 20.0 & 20.4   & 3.1 & 3.1 & 2.9 & 2.7 \\
&lss-qtl & 92.6 & 91.3 & 89.5 & 79.1   & 92.3 & 90.1 & 88.0 & 86.4   & 92.9 & 92.2 & 86.7 & 85.7   & 4.7 & 4.8 & 4.1 & 4.6   & 16.7 & 16.5 & 16.8 & 16.3   & 3.1 & 3.2 & 3.2 & 3.0 \\
&lss-clt & 86.1 & 71.4 & 72.8 & 29.3   & 77.7 & 66.6 & 59.9 & 55.2   & 85.6 & 72.8 & 47.5 & 44.9   & 4.9 & 5.2 & 4.5 & 11.1  & 19.3 & 19.3 & 22.2 & 22.9   & 3.3 & 3.5 & 6.2 & 5.9 \\[.2cm]
\cmidrule{1-26}
\end{tabular*}

%% file: Tables/TAB_SIM_20_HD_uniform.tex
\begin{tabular*}{\textwidth}{@{\extracolsep{\fill}} |llrrrr|rrrr|rrrr|rrrr|rrrr|rrrr|}
\cmidrule{1-26}
\multirow{3}{*}{\rotatebox[origin=c]{90}{nominal}}&$\beta$& \multicolumn{12}{c|}{Coverage}& \multicolumn{12}{|c|}{Winkler loss}\\
&$(e_i)$& \multicolumn{4}{c}{short-heavy}& \multicolumn{4}{c}{long-heavy}&\multicolumn{4}{c|}{non-lin-heavy}& \multicolumn{4}{|c}{short-heavy}& \multicolumn{4}{c}{long-heavy}&\multicolumn{4}{c|}{non-lin-heavy}\\
&$m$-days& 1 & 2 & 3 & 4 & 1 & 2 & 3 & 4  & 1 & 2 & 3 & 4  & 1 & 2 & 3 & 4 & 1 & 2 & 3 & 4 & 1 & 2 & 3 & 4 \\
\cmidrule{2-26}
\multirow{3}{*}{\rotatebox[origin=c]{90}{$60\%$}}
&lss-qtl & 51.7 & 48.9 & 46.8 & 42.2   & 51.9 & 47.3 & 44.9 & 39.3   & 50.4 & 48.6 & 46.3 & 40.7   & 4.8 & 3.8 & 3.9 & 3.4   & 8.6 & 7.7 & 7.8 & 7.8   & 3.1 & 2.5 & 2.5 & 2.3 \\
&lss-clt & 54.9 & 39.5 & 37.0 & 32.5   & 68.6 & 49.1 & 42.0 & 34.7   & 48.6 & 34.8 & 33.4 & 27.6   & 5.5 & 4.9 & 4.9 & 4.3   & 8.0 & 7.1 & 7.3 & 7.5   & 4.0 & 3.6 & 3.6 & 3.5 \\
&lss-adj & 51.3 & 47.0 & 47.9 & 47.2   & 50.4 & 46.0 & 44.2 & 40.6   & 49.5 & 48.9 & 48.8 & 48.7   & 4.8 & 3.9 & 3.8 & 3.2   & 8.4 & 7.2 & 6.9 & 6.8   & 3.2 & 2.5 & 2.3 & 2.1 \\[.2cm]
\multirow{3}{*}{\rotatebox[origin=c]{90}{$80\%$}}
&lss-qtl & 72.1 & 66.2 & 61.8 & 54.3   & 70.7 & 64.2 & 56.9 & 50.5   & 70.7 & 66.3 & 62.1 & 55.1   & 7.4 & 6.2 & 5.5 & 5.0   & 13.5 & 12.0 & 11.8 & 11.7   & 4.7 & 4.1 & 3.6 & 3.3 \\
&lss-clt & 74.5 & 57.7 & 51.8 & 49.2   & 82.1 & 67.7 & 59.5 & 51.2   & 69.4 & 53.8 & 46.6 & 40.5   & 7.8 & 6.7 & 7.1 & 6.3   & 12.5 & 10.5 & 10.8 & 11.2   & 5.3 & 5.0 & 5.2 & 5.3 \\
&lss-adj & 71.0 & 64.6 & 64.6 & 64.6   & 69.7 & 64.2 & 57.7 & 55.2   & 70.4 & 67.7 & 67.0 & 64.1   & 7.5 & 6.1 & 5.7 & 4.8   & 12.4 & 10.9 & 10.7 & 10.5   & 4.9 & 4.0 & 3.5 & 3.1 \\[.2cm]
\multirow{3}{*}{\rotatebox[origin=c]{90}{$90\%$}}
&lss-qtl & 82.8 & 74.3 & 68.5 & 61.0   & 81.1 & 73.1 & 62.6 & 55.4   & 81.5 & 74.2 & 67.9 & 62.9   & 12.4 & 9.0 & 8.2 & 7.8  & 20.7 & 17.4 & 17.9 & 18.7   & 8.0 & 6.0 & 5.2 & 5.0 \\
&lss-clt & 84.1 & 70.0 & 64.6 & 59.8   & 88.5 & 78.3 & 70.6 & 62.8   & 80.7 & 67.2 & 59.7 & 52.3   & 11.5 & 8.8 & 9.9 & 8.9  & 18.8 & 15.7 & 16.0 & 16.7   & 7.4 & 6.7 & 7.3 & 7.8 \\
&lss-adj & 81.3 & 76.8 & 74.2 & 71.7   & 80.4 & 73.7 & 68.6 & 65.0   & 80.4 & 78.1 & 76.3 & 73.7   & 12.0 & 8.7 & 8.1 & 7.2  & 19.4 & 16.6 & 16.3 & 16.5   & 7.5 & 5.8 & 5.0 & 4.5 \\[.2cm]
\multirow{3}{*}{\rotatebox[origin=c]{90}{$95\%$}}
&lss-qtl & 86.4 & 77.4 & 73.0 & 64.9   & 85.4 & 76.5 & 66.7 & 57.6   & 85.1 & 77.3 & 72.6 & 66.6   & 19.3 & 14.0 & 13.2 & 13.0   & 32.9 & 28.0 & 30.0 & 32.0   & 11.4 & 9.3 & 8.3 & 8.3 \\
&lss-clt & 90.4 & 80.0 & 74.0 & 70.2   & 92.1 & 83.2 & 77.3 & 70.3   & 86.2 & 77.2 & 69.6 & 62.8   & 17.5 & 12.4 & 13.6 & 12.4   & 29.1 & 24.2 & 24.9 & 26.1   & 10.6 & 9.1 & 10.3 & 11.5 \\
&lss-adj & 86.1 & 82.9 & 80.7 & 78.8   & 85.9 & 78.6 & 75.3 & 68.7   & 86.2 & 83.1 & 82.3 & 80.6   & 18.8 & 13.0 & 12.1 & 11.2   & 31.0 & 25.9 & 26.2 & 26.9   & 11.2 & 8.6 & 7.5 & 6.9 \\[.2cm]
\cmidrule{1-26}
\end{tabular*}